\documentclass[amscd,amssymb,11pt]{amsart}
\usepackage{amsmath,amsthm,amscd,amssymb}
\usepackage{latexsym}
\usepackage{color}
\usepackage{mathrsfs}
\usepackage[margin=1.55in]{geometry}
\usepackage{geometry}
\usepackage{graphicx}

\usepackage{tikz}
\usetikzlibrary{calc}

\usepackage{subfig}

\numberwithin{equation}{section}

\newtheorem{Thm}{Theorem}[section]

\newtheorem{Cor}[Thm]{Corollary}
\newtheorem{Lem}[Thm]{Lemma}
\newtheorem{Prop}[Thm]{Proposition}

\newtheorem{Rmk}[Thm]{Remark}

\newcommand{\p}{\varphi}

\newcommand{\N}{\mathbb{N}}

\newcommand{\R}{\mathbb{R}}

\renewcommand{\div}{{\rm div}}

\usepackage{etoolbox}

\makeatletter
\let\old@tocline\@tocline
\let\section@tocline\@tocline
\newcommand{\subsection@dotsep}{4.5}
\newcommand{\subsubsection@dotsep}{4.5}
\patchcmd{\@tocline}
  {\hfil}
  {\nobreak
     \leaders\hbox{$\m@th
        \mkern \subsection@dotsep mu\hbox{.}\mkern \subsection@dotsep mu$}\hfill
     \nobreak}{}{}
\let\subsection@tocline\@tocline
\let\@tocline\old@tocline

\patchcmd{\@tocline}
  {\hfil}
  {\nobreak
     \leaders\hbox{$\m@th
        \mkern \subsubsection@dotsep mu\hbox{.}\mkern \subsubsection@dotsep mu$}\hfill
     \nobreak}{}{}
\let\subsubsection@tocline\@tocline
\let\@tocline\old@tocline

\let\old@l@subsection\l@subsection
\let\old@l@subsubsection\l@subsubsection

\def\@tocwriteb#1#2#3{%
  \begingroup
    \@xp\def\csname #2@tocline\endcsname##1##2##3##4##5##6{%
      \ifnum##1>\c@tocdepth
      \else \sbox\z@{##5\let\indentlabel\@tochangmeasure##6}\fi}%
    \csname l@#2\endcsname{#1{\csname#2name\endcsname}{\@secnumber}{}}%
  \endgroup
  \addcontentsline{toc}{#2}%
    {\protect#1{\csname#2name\endcsname}{\@secnumber}{#3}}}%

\newlength{\@tocsectionindent}
\newlength{\@tocsubsectionindent}
\newlength{\@tocsubsubsectionindent}
\newlength{\@tocsectionnumwidth}
\newlength{\@tocsubsectionnumwidth}
\newlength{\@tocsubsubsectionnumwidth}
\newcommand{\settocsectionnumwidth}[1]{\setlength{\@tocsectionnumwidth}{#1}}
\newcommand{\settocsubsectionnumwidth}[1]{\setlength{\@tocsubsectionnumwidth}{#1}}
\newcommand{\settocsubsubsectionnumwidth}[1]{\setlength{\@tocsubsubsectionnumwidth}{#1}}
\newcommand{\settocsectionindent}[1]{\setlength{\@tocsectionindent}{#1}}
\newcommand{\settocsubsectionindent}[1]{\setlength{\@tocsubsectionindent}{#1}}
\newcommand{\settocsubsubsectionindent}[1]{\setlength{\@tocsubsubsectionindent}{#1}}

\renewcommand{\l@section}{\section@tocline{1}{\@tocsectionvskip}{\@tocsectionindent}{}{\@tocsectionformat}}%
\renewcommand{\l@subsection}{\subsection@tocline{1}{\@tocsubsectionvskip}{\@tocsubsectionindent}{}{\@tocsubsectionformat}}%
\renewcommand{\l@subsubsection}{\subsubsection@tocline{1}{\@tocsubsubsectionvskip}{\@tocsubsubsectionindent}{}{\@tocsubsubsectionformat}}%
\newcommand{\@tocsectionformat}{}
\newcommand{\@tocsubsectionformat}{}
\newcommand{\@tocsubsubsectionformat}{}
\expandafter\def\csname toc@1format\endcsname{\@tocsectionformat}
\expandafter\def\csname toc@2format\endcsname{\@tocsubsectionformat}
\expandafter\def\csname toc@3format\endcsname{\@tocsubsubsectionformat}
\newcommand{\settocsectionformat}[1]{\renewcommand{\@tocsectionformat}{#1}}
\newcommand{\settocsubsectionformat}[1]{\renewcommand{\@tocsubsectionformat}{#1}}
\newcommand{\settocsubsubsectionformat}[1]{\renewcommand{\@tocsubsubsectionformat}{#1}}
\newlength{\@tocsectionvskip}
\newcommand{\settocsectionvskip}[1]{\setlength{\@tocsectionvskip}{#1}}
\newlength{\@tocsubsectionvskip}
\newcommand{\settocsubsectionvskip}[1]{\setlength{\@tocsubsectionvskip}{#1}}
\newlength{\@tocsubsubsectionvskip}
\newcommand{\settocsubsubsectionvskip}[1]{\setlength{\@tocsubsubsectionvskip}{#1}}

\patchcmd{\tocsection}{\indentlabel}{\makebox[\@tocsectionnumwidth][l]}{}{}
\patchcmd{\tocsubsection}{\indentlabel}{\makebox[\@tocsubsectionnumwidth][l]}{}{}
\patchcmd{\tocsubsubsection}{\indentlabel}{\makebox[\@tocsubsubsectionnumwidth][l]}{}{}

\newcommand{\@sectypepnumformat}{}
\renewcommand{\contentsline}[1]{%
  \expandafter\let\expandafter\@sectypepnumformat\csname @toc#1pnumformat\endcsname%
  \csname l@#1\endcsname}
\newcommand{\@tocsectionpnumformat}{}
\newcommand{\@tocsubsectionpnumformat}{}
\newcommand{\@tocsubsubsectionpnumformat}{}
\newcommand{\setsectionpnumformat}[1]{\renewcommand{\@tocsectionpnumformat}{#1}}
\newcommand{\setsubsectionpnumformat}[1]{\renewcommand{\@tocsubsectionpnumformat}{#1}}
\newcommand{\setsubsubsectionpnumformat}[1]{\renewcommand{\@tocsubsubsectionpnumformat}{#1}}
\renewcommand{\@tocpagenum}[1]{%
  \hfill {\mdseries\@sectypepnumformat #1}}

\let\oldappendix\appendix
\renewcommand{\appendix}{%
  \leavevmode\oldappendix%
  \addtocontents{toc}{%
    \protect\settowidth{\protect\@tocsectionnumwidth}{\protect\@tocsectionformat\sectionname\space}%
    \protect\addtolength{\protect\@tocsectionnumwidth}{2em}}%
}
\makeatother

\makeatletter
\settocsectionnumwidth{2em}
\settocsubsectionnumwidth{2.5em}
\settocsubsubsectionnumwidth{3em}
\settocsectionindent{1pc}%
\settocsubsectionindent{\dimexpr\@tocsectionindent+\@tocsectionnumwidth}%
\settocsubsubsectionindent{\dimexpr\@tocsubsectionindent+\@tocsubsectionnumwidth}%
\makeatother

\settocsectionvskip{10pt}
\settocsubsectionvskip{0pt}
\settocsubsubsectionvskip{0pt}

\let\oldtableofcontents\tableofcontents
\renewcommand{\tableofcontents}{%
  \vspace*{-\linespacing}
  \oldtableofcontents}

\date{\today}

\title[Macroscopic Non-uniqueness]{
Macroscopic Non-uniqueness and Limits of Hamiltonian Dynamics}

\author{S. Dostoglou}
\author{Jianfei Xue}

\address{Department of Mathematics, 
          University of Missouri, 
          Columbia, MO 65211}

\begin{document}

\begin{abstract}
{We construct explicit examples of 
spontaneous energy generation and
non-uniqueness for the compressible Euler system, with and without pressure, by taking limits of Hamiltonian dynamics as the number of molecules increases to infinity. 
The examples come from rescalings of well-posed, deterministic systems of molecules that either collide elastically or interact via singular pair potentials.}
\end{abstract}

         
\maketitle
\tableofcontents

\section{Introduction}
{
Non-uniqueness for weak solutions of hydrodynamic equations is well known. 
Examples include the
construction of V.\,Scheffer \cite{Sch} and A.\,Shnirelman \cite{Sh} of non-trivial weak solutions of (incompressible, two-dimensional) Euler equations with compact time and space support,
and the work by C.\,De\,Lellis and L.\,Sz\'ekelyhidi \cite{dLS} showing that non-uniqueness (of 
the incompressible and compressible Euler equations in dimension greater or equal to two) persists even under ``admissibility" conditions. \cite{D} is a standard reference on the non-uniqueness of weak solutions of hyperbolic conservation laws in general.

In an attempt to investigate the origin of this behavior, we adopt here the point of view that hydrodynamic equations are the result of averaging  microscopic evolution equations (cf.\,\cite{M}, p.\,81, and \cite{B}, Part I, \S 20) 
to construct explicit examples of 
spontaneous macroscopic energy generation and
non-uniqueness for the compressible Euler system, with and without pressure, as limits of Hamiltonian dynamics. 
Our examples are rescaled limits of well-posed, deterministic systems of molecules that either collide elastically or interact via rescaled, singular pair interaction potentials, at the limit of infinitely many molecules,
cf.\,C.B.\,Morrey's work \cite{Mor}. 
For each moment $t$ and finite $N$, the positions and velocities of the molecules define the probability measure $\displaystyle M_t^{(N)}(d\bold x,d\bold v):=\dfrac 1N \sum_{k=1}^N\delta_{(\bold x_k,\bold v_k)}(d\bold x,d\bold v)$. In all examples here, $M_t^{(N)}$ converges weakly to $M_t$ as $N \to \infty$ and for each $(t,\bold x)$ the macroscopic density is given by the first marginal of $M_t$ and the macroscopic velocity by the barycentric projection of $M_t$ at $\bold x$ with respect to this marginal.

The first part of this article, consisting of Sections \ref{Section:Gost} and
\ref{section:reversal}, is centered on an example showing spontaneous generation of macroscopic velocitiy.
 The microscopic systems start with groups of motionless molecules and a single molecule, macroscopically undetectable, initially at a sufficiently large distance from the group, moving towards the group. Macroscopically, the limit of these flows describes a line segment in $\R^2$ at rest for $t\in(-\infty,0]$, which splits into two equal parts moving away from each other with velocities $\pm 1$ as soon as $t$ becomes positive. The macroscopic velocity and the macroscopic density from $M_t$ turn out to be a weak solution of the $2$-dimensional presureless Euler for all $t$ in $\R$. This solution is macroscopically as ``inadmissible" as those of Scheffer and  Shnirelman in that kinetic energy is spontaneously created at $t_0$. (Microscopically, total energy is, of course, conserved.) 
As Hamiltonian flows are time reversible, in Section \ref{section:reversal} the flows $M_t^{(N)}$ are reversed to produce a solution to the $2$-dimensional presureless Euler that does decrease energy. When compared to an elementary transverse flow, this provides an example of non-uniqueness for the presureless Euler under the admissibility condition of non-increasing energy.

In the second part of this article,  Section  \ref{layers} provides an interpretation, via a microscopic derivation, of the well known non-uniqueness of the Cauchy problem for the $1$-dimensional Euler system. We show how three moment equations derived from the transport equation
 \begin{equation}\label {eq: free flow eqn intro}
 \partial_t M_t + v\partial_xM_t =0
 \end{equation} 
can result in the $1$-dimensional Euler system.
The main point here is that two flows of probability measures solving the same transport equation, even if their moments coincide at $t=0$, in general will not have identical moments for all later times. Indeed, we construct two limit measures $M_t$ and $\widetilde  M_t$ both solving \eqref {eq: free flow eqn intro} and 
resulting in the $1$-dimensional Euler system.
At $t=0$, both $ M_t$ and $\widetilde M_t$ give the same macroscopic density, velocity, and pressure. Macroscopically, the solutions produced by $M_t$ and $\widetilde  M_t$ can be pictured as a segment of two and three layers, respectively,   on top of each other moving freely, see Figures \ref{Fig: two vertical leayers} and \ref{Fig: three vertical leayers}. The solutions in Section \ref{layers} are surrounded by vacuum (zero density).

}

\section{Preliminaries and Notation}
\label{why transport}
\subsection{Measure theory}

Recall that a sequence of finite measures $M_n(dx)$ converges {\em weakly} to a finite measure $M(dx)$ if for any $f(x)$ continuous and bounded $\displaystyle \int f(x) M_n(dx) \to \int f(x) M(dx)$, $n \to \infty$.
We then write $M_n \Rightarrow M$.

\medskip
\noindent
For $f:X\to Y$ measurable and $M$ a probability measure on $X$ 
the {\em push-forward measure} $f_\# M$  of $f$ on $Y$ (the distribution measure of the random variable $f$) is $(f_\# M) (B) =  M(f^{-1}(B))$. We often write $fM$ for this push-forward.
 
\medskip
\noindent
If $M$ is on $\R^{2d}$ its {\em first marginal} will be $(\pi_1)_\# M$, for $\pi_1: \R^{2d} \to \R^{d}, \pi_1(\bold x,\bold v) = \bold x$. 

\medskip
\noindent
$\displaystyle \int M_{\bold x}(d\bold v) \mu(d\bold x)$ is a shorthand for the measure 
$f \mapsto \displaystyle \int \left(\int f(\bold x,\bold v)M_{\bold x}(d\bold v)\right) \mu(d\bold x)$.

\medskip
\noindent
The {\em disintegration} of $M(d\bold x,d\bold v)$  with respect to its first marginal $\mu(d\bold x)$ is
 the unique, up to a $\mu$-measure $0$,
 family $M_{\bold x}(d\bold v)$  such that
$\displaystyle  M(d\bold x,d\bold v)=
     \int M_{\bold x}(d\bold v) \mu(d\bold x) $.
The barycentric projection of this disintegration is 
$ \displaystyle    \overline{\bold v}(\bold x) 
     = 
     \int \bold v M_{\bold x}(d\bold v)$ for $\bold x$ in the support of $\mu$, $\overline{\bold v} = 0$ otherwise. For details see \cite{AGS}, Section 5.3, or \cite{DJX}, Section 3.1. 
 
 \subsection{Finite systems}
 \label{subsec: finite}
A system of $N$ molecules in $\R^d$ will be described by the positions and velocities of the molecules, $\left(\bold x_k (t), \bold u_k (t) \right)$, $1 \leq k \leq N$, evolving via Hamiltonian dynamics with pairwise interaction $\Phi_{\sigma}(r)$ of finite range $\sigma$:
\begin{equation} \label{Ham}
\begin{split}
     \frac{d }{dt} \bold x_k(t)
    &=
    \bold u_k(t),
    \\ 
    \frac{d}{dt} \bold u_k (t)
    &=
    -
    \frac{1}{N}
    \sum_{\substack{j=1\\ j\neq k}}^{N}
    \Phi'_{\sigma}
    \left( 
   \left|\bold x_k(t) - \bold x_j(t)\right|
    \right)
    \frac{\bold x_k(t) - \bold x_j(t)}
     {\left|\bold x_k(t) - \bold x_j(t)\right|}.
\end{split}
\end{equation}
Following Morrey \cite{Mor}, we shall take $\displaystyle \Phi_{\sigma}(r) = \Phi\left(\frac{r}{\sigma}\right)$ for some $\Phi:(0, \infty) \to [0,\infty)$ satisfying:
\begin{equation} \label{Uinteraction}
   \lim_{r \to 0} \Phi (r) = + \infty, \quad
   \Phi' \leq 0, \quad \Phi'' \geq 0, \quad
   \Phi(r)  \neq 0 \Leftrightarrow  0< r<1.
\end{equation}

 For each $N$, suppose that a system $\displaystyle \left( \bold x_k^{(N)}(t), \bold u_k^{(N)}(t)\right)$, $k=1,\ldots,N$ evolves according to 
 \eqref{Ham}. Of central importance will be the corresponding $t$-family of probability measures on $\R^{2d} $:
\begin{equation} \label {M_t^{N}}
\begin{split}
     M_t^{(N)}(d\bold x,d\bold v)&:=
     \frac1N
     \sum_{k=1}^{N} 
     \delta_{\left(\bold x_k^{(N)}(t), \bold u_k^{(N)}(t)\right)}
     (d \bold x,d \bold v),\ t\geq 0 \text{, or } t\in \R. 
\end{split}
\end{equation}
When $M_t^{(N)}$ converges weakly to some $M_t$, it is crucial to note that the empirical measure formed by neglecting a single molecule converges weakly to the same $M_t$. (In fact, neglecting $o(N)$ number of molecules has the same effect.)
In this sense, any single molecule is macroscopically {\em invisible}. The construction in Section \ref{Section:Gost} relies heavily on this observation.


\section{Spontaneous Macroscopic Velocity Generation from Hamiltonian Dynamics}
\label{Section:Gost}
This section presents an example of a microscopic Hamiltonian flow with macroscopic limit, as $N\to \infty$, that shows spontaneous velocity generation. 
The microscopic systems start with groups of motionless molecules and a single molecule, initially at a sufficiently large distance from the group, moving towards the group with large velocity. For $t<0$, as $N\to \infty$, the moving  molecule is invisible and the macroscopic system is motionless. However, as the moving molecule starts interacting with the group at $t=0$, its energy is transferred to the rest of the system in such a way that all other molecules acquire speed $1$ to create macroscopic velocity for $t>0$.

There are similarities here with Lanford \cite{L}, pp.\,50--53, although Lanford works with an infinite system of hard balls that always remains discrete, rather than the limit of finite Hamiltonian systems with interaction, and he does not obtain hydrodynamic equations.

Throughout this section we use $\mathbf Q_t$ for the segment
\begin{equation}
\displaystyle \left\{ (x,y): 0\leq x\leq 1, y=t \right\} \subset \R^2
\end{equation}
and  $\Delta_t(d\bold x)$ for the normalized $1$-dimensional Lebesgue measure on $\mathbf Q_t$.
\begin{Thm}\label{Thm: ghost example}
For each $N\in \N$, there exists $\sigma_N>0$ and $\left( \bold x_k^{(N)}(t), \bold u_k^{(N)}(t)\right)$, $k=1,\ldots,N$, solution of the Hamiltonian system \eqref{Ham} with interaction $\Phi_{\sigma_N}$ for all $t\in \R$, such that 
for all $t\in\R$, the sequence of empirical measures 
\begin{equation}
\begin{split}
     M_t^{(N)}(d\bold x,d\bold v):=
     \frac1N
     \sum_{k=1}^{N} 
     \delta_{\left(\bold x_k^{(N)}(t), \bold u_k^{(N)} (t)\right)}
     (d \bold x,d \bold v) 
\end{split}
\end{equation}
converges weakly, as $N\to \infty$, to
\begin{equation}\label{eq: M_t ghost for all t}
M_t(d\bold x, d \bold v)
=
\left\{
\begin{array}{ll}
         \Delta_0(d\bold x)
         \otimes
         \delta_{\left(0,0\right)}\left(d\bold v\right)
       &
        t\leq 0 
        \\
         \dfrac12
         \Delta_{t}(d\bold x)
         \otimes
         \delta_{\left(0,1\right)}\left(d\bold v\right)
         +
         \dfrac 12
         \Delta_{-t}(d\bold x)
         \otimes
         \delta_{\left(0,-1\right)}\left(d\bold v\right)
      &
      t>0.
\end{array}
\right.
\end{equation} 
\end{Thm}
The proof of this theorem occupies the rest of this section. For the moment, note that
the first marginal (macroscopic density) of $M_t(d\bold x, d\bold v)$ in \eqref{eq: M_t ghost for all t} is 
\begin{equation} \label{eq: macro density from ghost}
\begin{split}
     \mu_t(d\bold x) 
     =
     \left\{
     \begin{array}{ll}
      \Delta_0(d\bold x) 
     & t\leq 0
     \\
     \dfrac12 \Delta_t(d\bold x)  + \dfrac 12 \Delta_{-t}(d\bold x)
     & t> 0 ,
     \end{array}
     \right. 
\end{split}
\end{equation}
If we disintegrate 
\begin{equation}\label {eq: disintegration of ghost Mt}
 M_t(d\bold x,d\bold v)
=
\int M_{t,\bold x}(d\bold v) \, \mu_t(d\bold x)
\end{equation}
 then
\begin{equation}\label {eq: M_t,x}
\begin{split}
     M_{t,\bold x}(d\bold v) 
     =
     \left\{
     \begin{array}{ll}
      \chi_{\mathbf Q_0}(\bold x)    \delta_{(0,0)}(d\bold v) 
     & t\leq 0
     \\
     \chi_{\mathbf Q_t}(\bold x) \delta_{(0,1)} (d\bold v)  + \chi_{\mathbf Q_{-t}}(\bold x) \delta_{(0,-1)} (d\bold v)
     & t> 0 .
     \end{array}
     \right. 
\end{split}
\end{equation}
Notice that in \eqref{eq: M_t,x} we only needed to specify $M_{t,\bold x}(d\bold v)$ for $\bold x$ in the support of $\mu_t(d\bold x)$.
The macroscopic velocity is the barycentric projection of this disintegration:
\begin{equation}\label{eq: macro velocity from ghost}
\begin{split}
   \bold u(t,\bold x)
   :=
   \int_{\R^2} \bold v M_{t,\bold x}(d\bold v)
   =
     \left\{
     \begin{array}{cl}
     (0,0)  &t\leq 0\\
    \chi_{\mathbf Q_t}(\bold x) \cdot (0,1) +  \chi_{\mathbf Q_{-t}}(\bold x) \cdot (0,-1)
     &t> 0.
     \end{array}
     \right.
     \end{split}
\end{equation}
The macroscopic density \eqref{eq: macro density from ghost} and velocity \eqref{eq: macro velocity from ghost} show clearly a macroscopic velocity generation (see Figure \ref{Ghost Flow}):  before $t=0$, the macroscopic system stays at rest, while, starting at $t=0$, two equal mass fronts split and move away from each other with velocity $\pm 1$.
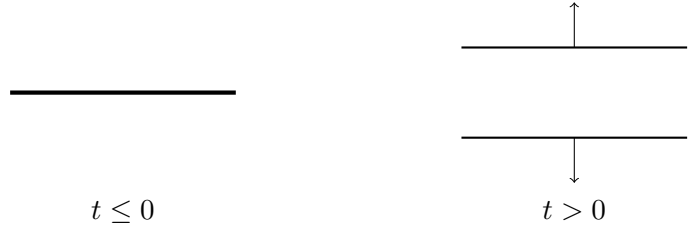
\begin{figure}
\begin{tikzpicture}
  \coordinate []  (A1) at (0,0);
  \coordinate [] (A2) at (3,0);
  \coordinate [] (B1) at (6,0.6);
  \coordinate [] (B2) at (9,0.6);
  \coordinate [] (C1) at (6,-0.6);
  \coordinate [] (C2) at (9,-0.6);
    \coordinate [] (M1) at (7.5,0.6);
  \coordinate [] (M2) at (7.5,1.2);
    \coordinate [] (N1) at (7.5,-0.6);
  \coordinate [] (N2) at (7.5,-1.2);
  \draw [ultra thick](A1) -- (A2);
  \draw [thick](B1) -- (B2);
  \draw [thick](C1) -- (C2);
  \draw [->](M1) -- (M2);
  \draw [->](N1) -- (N2);
  \coordinate [label=below:$t \leq 0$] (v) at (1.5,-1.3);
  \coordinate [label=below:$t > 0$] (v) at (7.5,-1.3);
   \end{tikzpicture}
    \caption{Macroscopic flow of $M_t$ in Theorem \ref{Thm: ghost example}.}
    \label{Ghost Flow}
\end{figure}
The sudden increase of macroscopic kinetic energy, of course, comes from interaction with an invisible molecule as we will see in the proof of Theorem \ref{Thm: ghost example} (subsections \ref{subsection: 2 particle system},  \ref{subsection: a system of molecules}, and \ref{subsection: N to infty}). 
In subsection \ref{Macroscopic Equations} we examine the macroscopic  hydrodynamic equation solved by the density \eqref{eq: macro density from ghost} and velocity \eqref{eq: macro velocity from ghost}.

\subsection{Interaction with one particle at rest}\label{subsection: 2 particle system}
\begin{figure} 
\begin{tikzpicture}
\draw   [dotted] (4,-4)  -- (4,3);
\draw   (0,0) circle [radius=1.5]; 
\coordinate [label=left:${(x_0,y_0)}$]  (v) at (0,0);
\node at (0,0)  {.};
\node at (.772,2.059)  {.};
\node at (-.772,-2.059)  {.};
\draw [dotted]  (.772,2.059)--(.772+4,2.059-1.5);
\draw [dotted]  (-.772,-2.059)--(-.772+1.3*4,-2.059-1.3*1.5); 
\node at (4,2.059-1.5*.807)  {.};
\node at (4,-2.059-1.5*1.193)  {.};
\draw  [|-|](4,2.059-1.5*.807)--(4,-2.059-1.5*1.193);
\coordinate [label=left:${I}$]  (I) at (3.6,-0.3-1.5*.807);
\draw  [dotted] (0,0) circle [radius=2.2];
\draw   [dotted] (0,0)  -- (4.,-1.5);
\draw   [<->] (0,3.)  -- (4.,3.);
\draw   [dotted] (0,0)  -- (1.5,0);
\coordinate [label=left:${\phi}$]  (phi) at (1.4,-.2);
\coordinate [label=above:${d}$]  (d) at (2,3.2);
\coordinate [label=right:${(x_0+d,y_0 + d \tan \phi)}$]  (v) at (4,-1.5);
\node at (4,-1.5)  {.};
\coordinate [label=above:${D}$]  (D) at (0,-1.);
\draw [solid] (0,0) -- (0,1.5);
\draw [solid] (0,0) -- (0,2.2);
\coordinate [label=left:${r}$]  (r) at (.5,.8) ;
\coordinate [label=left:$\sigma$]  (sigma) at (.5,1.8) ;

\end{tikzpicture}
\caption{The initial disc $D$ and the segment $I$ for $\phi<0$.}
\label {fig: the single step}
\end{figure}
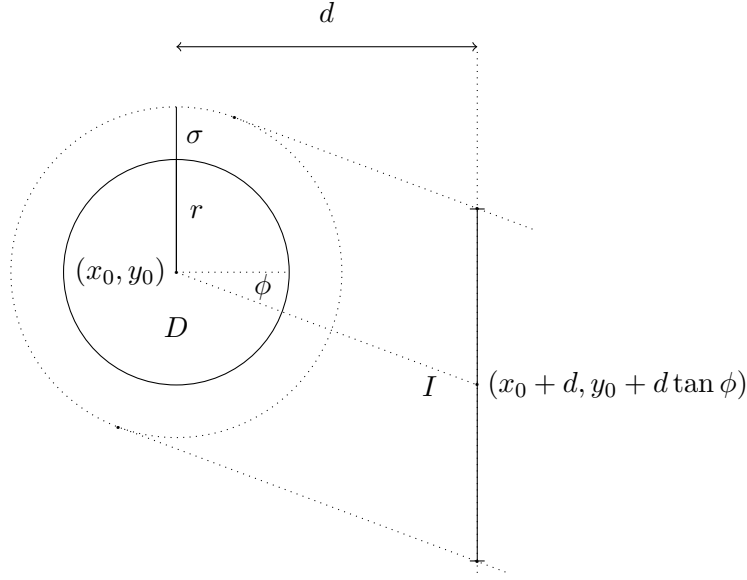
Start with two identical molecules $P$, $Q$ interacting with potential $\Phi_{\sigma}$ as in \eqref{Ham}.
Denote the positions and velocities of $P$, $Q$ as $\bold x_P= (x_P,y_P)$, $\bold x_Q = (x_Q,y_Q)$, $\bold v_P$, and $\bold v_Q$. Consulting Figure \ref{fig: the single step},
let $D$ be  the disc with center $(x_0,y_0)$ and radius $r>0$ and assume that at $t=0$ 
\begin{enumerate}
\item $(x_P,y_P)\in D$ and $x_Q=x_0+d$ with $d>r+\sigma$, i.e. $P$ is inside $D$ and $Q$ is on the vertical  line $x=x_0+d$.
\item $\bold v_P=v(\cos \phi,\sin \phi)$ with $-\dfrac{\pi}{2}<\phi <\dfrac{\pi}{2}$, $v>0$ and $\bold v_Q=(0,0)$, i.e. $P$ moves with speed $v$ and $Q$ is at rest.
\end{enumerate}

We say that there is interaction between $P$ and $Q$ whenever their distance is smaller than $\sigma$. Since $Q$ is at rest at $t=0$, there will be no interaction between $P$ and $Q$ as long as $P$ is inside $D$. The following lemma on the interaction between  $P$ and $Q$ is the building block of the rest of this section.
\begin{Lem}\label{successive collision}
Let  $P$, $Q$ be as above:
\begin{enumerate}
\item \label{second item}
For any $\theta$ in  $\left(-\dfrac{\pi}{2},\dfrac{\pi}{2}\right)$ there exists $y_Q$ such that $P$ and $Q$ will eventually interact (i.e. $P$ and $Q$ will interact at some time $t>0$), and after interaction $P$  and $Q$ will move in directions perpendicular to each other with constant velocities  $\bold v_P'=v\cos\theta\left(\cos\phi',\sin\phi'\right)$ and $\bold v_Q'=v\sin\theta\left(\sin\phi',-\cos\phi'\right)$, respectively, where $\phi'=\phi+\theta$.
\item
If interaction takes place then $y_Q$ satisfies
 \begin{equation*}
\begin{split}
     \left|\, y_Q-(y_0+ d\tan\phi)\,\right|
     <
     \dfrac{r+\sigma}{\cos\phi}.
\end{split}
\end{equation*}
 \item 
Whenever $P$ and $Q$ interact, they are both inside the disc with center $(x_0+d,y_0+d\tan\phi)$ and radius  $\dfrac{r+\sigma}{\cos\phi}+ 5 \sigma$.
 \end{enumerate}
\end{Lem}
\begin{proof}
(1) 
Consulting Figure \ref{Collision triangle} (which is \cite{LL}'s Figure 17, p.\,47, in our notation),
for $\theta$ the deflection angle from $\bold v_P$ to $\bold v_{P}'$, conservation of momentum and energy gives the formulas of $\bold v_P'$ and $\bold v_Q'$. 
\begin{figure}
\begin{tikzpicture}
  \coordinate []  (0) at (-1,0);
  \coordinate [] (A) at (-4,0.);
  \coordinate [] (B) at (2.0, 0.);
  \coordinate [] (C) at (1.24, 2.);
  \draw [->,very thick](A) -- (C);
  \draw [->,very thick](A) -- (B);
  \draw [->,very thick](C) -- (B);
   \coordinate [label=below:$\bold v_P$] (v) at (-0.7,0.);
  \coordinate [label=below:$\bold v'_P$] (v) at (-1.,2.24);
  \coordinate [label=below:$\bold v'_Q$] (v) at (2.4,1.5);
  \draw[->,very thick]  (-2.5,0.)  arc (0:20:15mm);
  \coordinate [label=above:$\ \ \theta$] (theta) at (-2.2,0);
  \draw (1.24-5.6*0.22/6 , 2-2.14*0.22/6) -- (1.24-3.46*0.22/6,2-7.74*0.22/6) -- (1.24+2.14*0.22/6,2-5.6*0.22/6);
  \coordinate []  (00) at (1,0);
  \coordinate [] (AA) at (3.,0.);
  \coordinate [] (BB) at (9.0, 0.);
  \coordinate [] (CC) at (8.24, -2.);
  \draw [->,very thick](AA) -- (CC);
  \draw [->,very thick](AA) -- (BB);
  \draw [->,very thick](CC) -- (BB);
  \coordinate [label=below:$\bold v_P$] (v) at (7.3,0.);
  \coordinate [label=below:$\bold v'_P$] (v) at (6.,-1.2);
  \coordinate [label=below:$\bold v'_Q$] (v) at (9,-.9);
 \draw[->,very thick]  (4.5,0.)  arc (0:-20:15mm);
  \coordinate [label=right:$\ \ \theta$] (theta) at (4.4,-.3);
 \draw (8.24-5.6*0.22/6 , -2+2.14*0.22/6) -- (8.24-3.46*0.22/6,-2+7.74*0.22/6) -- (8.24+2.14*0.22/6,-2+5.6*0.22/6);
   \end{tikzpicture}
    \caption{The two possible deflection triangles for
    given $|\bold v'_Q|$.}\label{Collision triangle}
\end{figure}
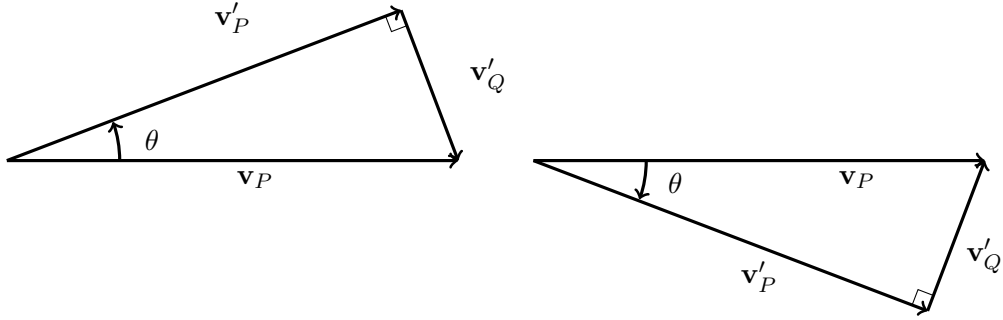
That any $\theta$ in $\left(-\dfrac{\pi}{2},\dfrac{\pi}{2}\right)$ is attained by some $y_Q$ follows from Corollary \ref{continuous dependence} in the Appendix and the formulas in \cite{LL}, \S 13 that show how to transform from motion in a central field to a system of two molecules. 

(2) Let $S$ be the strip between the two lines tangent to $D$ and parallel to $\bold v_P$, $S_{\sigma}$ the set of all points with distance smaller than $\sigma$ from $S$, and $I$ be the interval of intersection of  $S_{\sigma}$  with the line $x=x_0+d$, see Figure  \ref{fig: the single step}. Then if $Q$ has second coordinate anywhere out of $I$, $P$ ignores it and continues with unaltered velocity ${\bold v}_P$. Elementary geometry shows that $I$ has midpoint $y_0+d \tan\phi$ and half-length $\dfrac{r+\sigma}{\cos\phi}$, consult Figure \ref{fig: the single step}.

(3)
By Lemma \ref{Lem: CollisionTimeEstimate} in the Appendix, when $P$ and $Q$ interact, their interaction time is less than $\dfrac {4\sigma}{v}$ and by conservation of energy ($\Phi$ is positive) the speed of $Q$ will never be more than $v$ during interaction. 
Therefore, during interaction $Q$ travels less than $4 \sigma$, i.e.\ it stays in the disc centered at $\left(x_0+d,y_0+ d \tan \phi\right)$ with radius $\dfrac{r+\sigma}{\cos\phi}+ 4 \sigma$. 
As the distance between $P$ and $Q$ is always less than $\sigma$ during interaction, $P$ is always inside the circle centered at $\left(x_0+d,y_0+ d \tan \phi\right)$ with radius $\dfrac{r+\sigma}{\cos\phi}+ 5 \sigma$.
\end{proof}

\subsection{A system of molecules on the plane}\label{subsection: a system of molecules}
We describe now a system consisting of $N+1$ molecules $P, Q_k$, $k=1,\ldots,N$ where $P$ interacts (only once) with each $Q_k$ (in the order of increasing $k$) and interactions are independent ($P$ does not interact with $P_j$, $j\neq k$, when interacting with $P_k$, and there is no interaction between the $Q_k$'s). In addition, the moment before interacting with $Q_k$ the speed of $P$ will be greater than $1$ and the speed of $Q_k$ after interaction will be $1$.

We use $\theta_k$ for the deflection angle of $P$ due to the interaction with $Q_k$. Assume that before interacting with $Q_1$, $P$ moves along the $x$-axis. Then $ \displaystyle \phi_k=\sum_{j=1}^k\theta_j$ will be the angle from the $x$-axis to the direction of the velocity of $P$ right after its interaction with $Q_k$. The angle from the $x$-axis to the direction of the velocity of $Q_k$ after its interaction with $P$ will be denoted by $\hat \phi_k$. By Figure \ref{Collision triangle}, $\hat{\phi}_k =(-1)^{k+1}  \dfrac{\pi}{2} + \phi_k$.

\begin{Lem}\label{theta&phi}
For $N \in \N$ fixed and $k=1,2,\ldots,N$, let 
\begin{equation} \label{eq:theta&phi}
\begin{split}
     \theta_k=(-1)^k\arcsin\dfrac{1}{\sqrt{N+2-k}}, \quad
     \displaystyle\phi_k=\sum_{j=1}^k\theta_j.
\end{split}
\end{equation} 
Then 
\begin{enumerate}
\item $\phi_k<0$, when $k$ is odd and $\phi_k>0$, when $k$ is even,
\item $\left|\phi_k\right|<\left|\phi_{k+2}\right|$,
\item $\left|\phi_1\right| = \left| \theta_1 \right| \leq \dfrac{\pi}4$ and $\left|\phi_k\right| < \left| \theta_k \right| \leq \dfrac{\pi}4$,  for $k>1$.
\end{enumerate}
\end{Lem}
\begin{proof}[Proof of Lemma  \ref{theta&phi}] 
For (1), observe that the $\theta_k$'s start negative, increase strictly in absolute value and alternate sign. Therefore for $k$ odd and $k>1$
\begin{equation}
    \begin{split}
          \phi_{k}
          =
          \theta_1
          +\left( \theta_2+\theta_3 \right)
          +\ldots
          +\left(  \theta_{k-1}+\theta_k \right)
          <
          \theta_1<0,
    \end{split}
\end{equation}
whereas for $k$ even
\begin{equation}
    \begin{split}
          \phi_{k}
          =
          \left( \theta_1+\theta_2 \right)
          +\ldots
          +\left( \theta_{k-1}+\theta_k \right)
          >0.
    \end{split}
\end{equation}
For (2), notice that 
$\theta_{k+1}+\theta_{k+2}$ always has the same sign as $\phi_k$, hence 
\begin{equation}
    \begin{split}
        \left|\phi_{k+2}\right|
        =
        \left|\phi_k\right|+\left|\theta_{k+1}+\theta_{k+2}\right|>\left|\phi_k\right|.
    \end{split}
\end{equation}
$(3)$ For $1\leq k\leq N$
\begin{equation} \label{eq: estimate on theta}
    \begin{split}
         \left|\theta_k\right|=\arcsin\dfrac{1}{\sqrt{N+2-k}}\leq \arcsin\dfrac{1}{\sqrt 2}=\dfrac {\pi}{4}.
    \end{split}
\end{equation}
 For $k$ is odd and $k>1$
\begin{equation} \label{phi vs. theta odd}
    \begin{split}
          \left|\phi_{k}\right|=-\phi_{k}=-\phi_{k-1}-\theta_{k}<-\theta_{k}\leq\dfrac{\pi}{4},
    \end{split}
\end{equation}
whereas for $k$ even
\begin{equation} \label{phi vs. theta even}
          \left|\phi_{k}\right|=\phi_{k}=\phi_{k-1}+\theta_{k}<\theta_{k}\leq\dfrac{\pi}{4}. \qedhere
\end{equation}
\end{proof}
Lemma \ref{theta&phi} shows that the even $\phi_k$'s are positive, increasing, and never more than $\pi/4$ (and therefore the even $\hat\phi_k$'s are negative, increasing, and never more than $-\pi/4$), whereas the odd $\phi_k$'s are negative, decreasing, and never less than $-\pi/4$ (and therefore the odd $\hat\phi_k$'s are positive, decreasing, and never less than $\pi/4$).
Figure \ref{angles} summarizes the behavior of $\phi_k$ and $\hat\phi_k$.
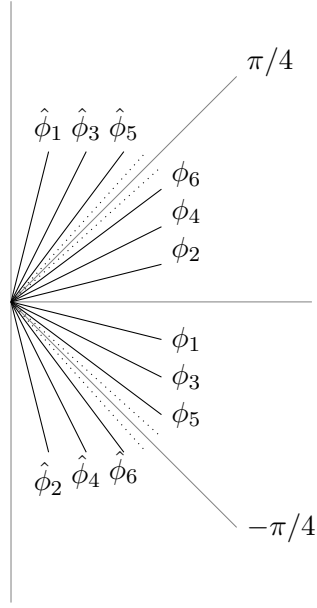
\begin{figure}
\begin{tikzpicture}
\draw   [help lines] (0,-4)  -- (0,4);
\draw   [help lines] (0,0)  -- (4,0);
\draw  [help lines] (0, 0)--(3,3); 
\coordinate [label=right:${\pi/4}$]  (pi/4) at (3,3.2);
\draw    (0, 0)--(2,1); 
\coordinate [label=right:${\phi_4}$]  (phi4) at (2,1.2);
\draw   (0, 0)--(2,1.5); 
\coordinate [label=right:${\phi_6}$]  (phi6) at (2,1.7);
\draw  (0, 0)--(2,.5); 
\coordinate [label=right:${\phi_2}$]  (phi2) at (2,.7);
\draw  [dotted] (0, 0)--(2,1.8);
\draw  [dotted] (0, 0)--(2,-1.8); 
\draw  [help lines] (0, 0)--(3,-3); 
\coordinate [label=right:${-\pi/4}$]  (-pi/4) at (3,-3.);
\draw   (0, 0)--(2,-1); 
\coordinate [label=right:${\phi_3}$]  (phi3) at (2,-1.);
\draw   (0, 0)--(2,-1.5); 
\coordinate [label=right:${\phi_5}$]  (phi5) at (2,-1.5);
\draw  (0, 0)--(2,-.5); 
\coordinate [label=right:${\phi_1}$]  (phi1) at (2,-.5);
\draw   (0, 0)--(1,-2); 
\coordinate [label=below:${\hat\phi_4}$]  (hatphi4) at (1,-1.9);
\draw   (0, 0)--(1.5, -2);
\draw   [dotted] (0, 0)--(1.8, -2);
\draw   [dotted] (0, 0)--(1.8, 2); 
\coordinate [label=below:${\hat\phi_6}$]  (hatphi6) at (1.5, -1.85);
\draw  (0, 0)--(.5, -2); 
\coordinate [label=below:${\hat\phi_2}$]  (hatphi2) at (.5, -2); 
\draw   (0, 0)--(1,2); 
\coordinate [label=above:${\hat\phi_3}$]  (hatphi3) at (1,2);
\draw   (0, 0)--(1.5,2); 
\coordinate [label=above:${\hat\phi_5}$]  (hatphi5) at (1.5,2);
\draw  (0, 0)--(.5,2); 
\coordinate [label=above:${\hat\phi_1}$]  (hatphi1) at (.5,2);
\end{tikzpicture}
\caption{The angles $\phi_k$ and $\hat{\phi}_k$. Observe how the even/odd $\phi_k$'s and the even/odd $\hat\phi_k$'s fall into four non-overlapping sectors.} \label{angles}
\end{figure} 

In the description of the interaction of $P$ and $Q_k$,
for $\phi_j$ as in \eqref{eq:theta&phi} (assuming $\phi_0 = 0$), the point
\begin{equation} \label{eq:x_k,y_k}
\begin{split}
     (x_k, y_k)
     =
     \left(\
     \dfrac kN, 
     \ \dfrac 1N \sum_{j=0}^{k-1}\tan\phi_j
     \ \right)
\end{split}
\end{equation} 
will play the same role as $(x_0+d , y_0+ d \tan \phi)$ in Lemma \ref{successive collision}. 
The segments and half-lines
\begin{equation} \label{eq:mathcalQ_n}
\begin{split}
      \mathcal{P}_k
        &:=
        \left\{  
        (x,y):
        x_k \leq x\leq x_{k+1},\
        y = \left(x-x_k\right)\tan \phi_k \,+y_k
        \right\}, \ k = 1, \ldots, N-1,
        \\
        \mathcal{P}_N
        &:=
        \left\{  
        (x,y):
        x_N \leq x,\
        y = \left(x-x_N\right)\tan \phi_N \,+y_N
        \right\},
\\
        \mathcal{Q}_k
        &:=
        \left\{
        (x,y): x\geq x_k,
        \
         y = (x-x_k)\tan \hat\phi_k\,+y_k  
        \right\}, \ k = 1, \ldots, N,
            \end{split}
\end{equation}
will be useful in describing the trajectories of $P$ and of each $Q_k$, respectively, see Figure \ref{Fig: P_nQ_n}. 
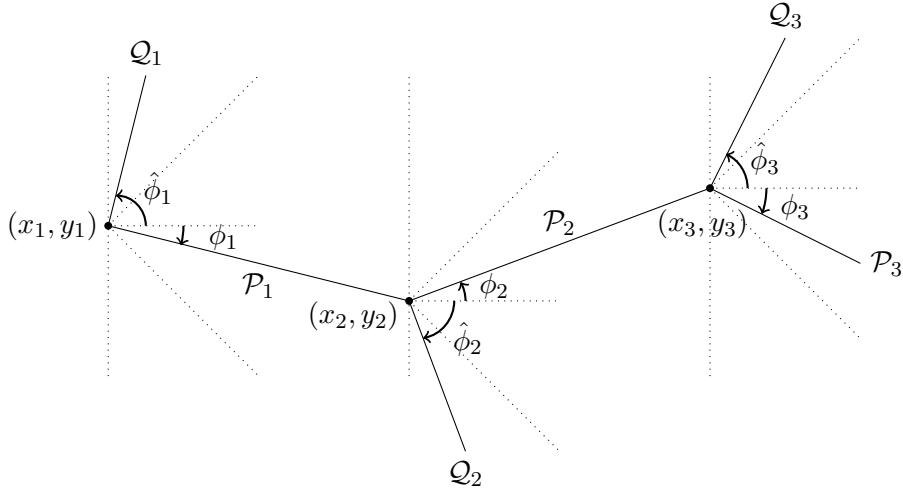
\begin{figure}
\begin{center}
\begin{tikzpicture}
\draw [dotted] (0,0) -- (0,4);
\draw [dotted] (4,0) -- (4,4);
\draw [dotted] (8,0) -- (8,4);
\draw[->, thick]  (1,2)  arc (0:-15:10mm);
\draw[->, thick]  (4.75,1)  arc (0:30:5mm);
\draw[->, thick]  (8.75,2.5)  arc (0:-20:10mm);
\draw[->, thick]  (0.5,2)  arc (0:90:4mm);
\draw[->, thick]   (4.6,1)  arc (0:-80:5mm);
\draw[->, thick]  (8.5,2.5)  arc (0:65:5mm);
\coordinate [label=right:${\phi_1}$]  (phi1) at (1.2,1.85);
\coordinate [label=right:${\phi_2}$]  (phi2) at (4.8,1.2);
\coordinate [label=right:${\phi_3}$]  (phi3) at (8.8,2.25);
\coordinate [label=left:${\hat\phi_1}$]  (hatphi1) at (1,2.5);
\coordinate [label=right:${\hat\phi_2}$]  (hatphi2) at (4.45,0.5);
\coordinate [label=right:${\hat\phi_3}$]  (hatphi3) at (8.4,2.9);
\draw [fill] (0,2) circle [radius=0.045];
\coordinate [label=left:${(x_1,y_1)}$]  (Q1) at (0,2);
\draw [fill] (4,1) circle [radius=0.045];
\coordinate [label=left:${(x_2,y_2)}$]  (Q2) at (4,0.8);
\draw [fill] (8,2.5) circle [radius=0.045];
\coordinate [label=below:${(x_3,y_3)}$]  (Q3) at (7.9,2.35);
\draw [dotted] (0,2) -- (2,4);
\draw [dotted] (0,2) -- (2,0);
\draw [dotted] (0,2) -- (2,2);
\draw [dotted] (4,1) -- (6,3);
\draw [dotted] (4,1) -- (6,1);
\draw [dotted] (4,1) -- (6,-1);
\draw [dotted] (8,2.5) -- (10,4.5);
\draw [dotted] (8,2.5) -- (10,2.5);
\draw [dotted] (8,2.5) -- (10,0.5);
\draw [solid] (0,2) -- (4,1) -- (8,2.5);
\coordinate [label=below:${\mathcal{P}_1}$]  (Q1) at (2,1.5);
\coordinate [label=above:${\mathcal{P}_2}$]  (Q2) at (6,1.75);
\draw [solid] (0,2) -- (0.5,4);
\coordinate [label=above:${\mathcal{Q}_1}$]  (P1) at (0.5,4);
\draw [solid] (4,1) -- (4.75,-1);
\coordinate [label=below:${\mathcal{Q}_2}$]  (P2) at (4.75,-1);
\draw [solid] (8,2.5) -- (9,4.5);
\coordinate [label=above:${\mathcal{Q}_3}$]  (P3) at (9,4.5);
\draw [solid] (8,2.5) -- (10,1.5);
\coordinate [label=right:${\mathcal{P}_3}$]  (Q3) at (10,1.5);
\end{tikzpicture}
\end{center}
\caption{First few segments $\mathcal{P}_n$ and half-lines $\mathcal{Q}_n$.}
\label{Fig: P_nQ_n}
\end{figure}
Define the distance between any two of these sets as 
\begin{equation}
\begin{split}
     d\left( \mathcal{A}, \mathcal{B}\right)
     :=
     \inf
      \left\{
      \|a -b \| :a\in\mathcal A,\ b\in \mathcal B
     \right\}.
\end{split}
\end{equation} 

\begin{Lem} \label{P_n Q_n}
Let $m,n=1,2,\ldots,N$.
For $\mathcal{Q}_m$, $\mathcal{P}_n$ as above, 
\begin{equation} \label{distance fm fn hn}
    \begin{split}
       d\left( \mathcal{Q}_m, \mathcal{Q}_n\right) &> \frac1N, \quad m \neq n,\\
       d\left( \mathcal{Q}_m, \mathcal{P}_n\right)&> \frac1N, \quad m < n.
    \end{split}
\end{equation}
\end{Lem}

\begin{proof}[Proof of Lemma \ref{P_n Q_n}] 
Recalling \eqref{eq:mathcalQ_n}, we use here ``right half plane of $\mathcal{Q}_n$"  to mean the half-plane to the right of the $y$-axis defined by: $\mathcal{Q}_n$ is the positive $y$-axis when $\hat\phi_n$ is positive; $\mathcal{Q}_n$ is the negative $y$-axis when $\hat\phi_n$ is negative.

Observe first that for any fixed $n$, the point $(x_n,y_n)$ is always in the right half plane of $\mathcal{Q}_m$ for all $m<n$ : this holds by the relation of the angles $\phi_i$ to the angles $\hat\phi_j$, see Figure \ref{angles} and Figure \ref{Fig: P_nQ_n}. 

To get the first estimate in \eqref{distance fm fn hn}, it suffices to consider $n>m$. If $n-m$ is even, then the angle of $\mathcal{Q}_n$ (i.e.\,$\hat \phi_n$) is of smaller absolute value than the angle of $\mathcal{Q}_m$ (i.e.\,$\hat \phi_m$). If $n-m$ is odd, then the angles of  $\mathcal{Q}_m$ and  $\mathcal{Q}_n$ differ by more than $\pi/2$. In either case the point on $\mathcal{Q}_n$ closest to $\mathcal{Q}_m$ is $(x_n, y_n)$.

Similarly, the angle of $\mathcal{P}_n$ (i.e.\,$\phi_n$), is always of absolute value smaller than the angle of any $\mathcal{Q}_m$  (i.e.\,$\hat \phi_m$). Therefore 
the point on $\mathcal{P}_n$  closest to $\mathcal{Q}_m$, for $m< n$, is $(x_n, y_n)$.

Now it suffices to notice that the distance from $(x_n,y_n)$ to each $Q_m$ is greater or equal to $|\mathcal P_m|$ which is clearly bigger than $\dfrac 1N$ (consult Figure \ref{angles} and Figure \ref{Fig: P_nQ_n}).
\end{proof}
We are now ready to establish the evolution of $P,Q_1,\ldots,Q_N$.
\begin{Prop} \label{Prop: The System}
For each $N\in \N$ and $\displaystyle \sigma_{N} < \frac{1}{2\sqrt{2}}\frac{1}{N(N+3)^{3/2}}$, consider the system $P,Q_1,\ldots,Q_N$ with interaction $\Phi_{\sigma_N}$. 
Then there exist $y_{Q_k}$'s, $k=1,\ldots,N$, such that the system evolves as follows: for $t\leq 0$, 
\begin{equation}  
\begin{split}\label{P_k initial setting}
     P(t)=t\bold v_P,\quad   &\bold v_P(t)=\left(\sqrt{N+1},0\right), \\
     Q_k(t)=\left(\dfrac kN,y_{Q_k}\right),\quad &\bold 
     v_{Q_k}(t)=(0,0), \quad k=1,\ldots, N,
\end{split}
\end{equation}
and for $t>0$,
\begin{enumerate}
\item
There exist times $0<t_1'<t_1''<t_2'<t_2''<\ldots<t_N'<t_N''$ such that
for any $1\leq k\leq N$, $P$ starts to interact with $Q_k$ at $t=t_k'$ and completes this interaction at $t=t_k''$.
\item
For any $1\leq k\leq N$, the molecule $Q_k$ does not interact with any other molecule for $t<t_k'$ or $t>t_k''$ and its velocity is given by
        \begin{equation} \label{v_{Q_k}}
        \begin{split}
        \bold v_{Q_k}(t)=
         \begin{cases}
        (0,0) &t<t_k'  \\
        \left(\sin{|\phi_k|},(-1)^{k+1}\cos{\phi_k}\right)  & t>t_k'',
        \end{cases}
        \end{split}
        \end{equation}
for $\phi_k$ as in \eqref{eq:theta&phi}. 
\item
The velocity of $P$ satisfies
       \begin{equation} \label{eq: velocity of P}
        \begin{split}
        \bold v_P(t)=
         \begin{cases}
         (\sqrt{N+1},0)&t \leq t_1' \\
       \sqrt{N+1-k}\left(\cos{\phi_k},\sin{\phi_k}\right) &t_{k}''\leq t \leq t_{k+1}', k=1,2,\ldots,N-1\\
        \left(\cos{\phi_N},\sin{\phi_N}\right) &t \geq t_N'',
        \end{cases}
        \end{split}
        \end{equation}
for $\phi_k$ as in \eqref{eq:theta&phi}.
\item
During the time interval $[t_k', t_k'']$ for $1\leq k \leq N$ the molecules $P$ and $Q_k$ are in the disc of center $(x_k,y_k)$ as in \eqref{eq:x_k,y_k} and radius given recursively by 
\begin{equation}\label{r_k r_{k-1}}
\begin{split}
       r_k = \dfrac{r_{k-1}+\sigma_{N}}{\cos \phi_{k-1}} + 5 \sigma_{N},
       \quad
        r_0 =0.
\end{split}
\end{equation}
In particular, 
\begin{equation}\label{estimate of r_k}
   \begin{split}
r_k < 2\sqrt 2 (N+3)^{3/2}\sigma_N.
   \end{split}
\end{equation}
\end{enumerate}
\end{Prop}
\begin{proof}[Proof of Proposition \ref{Prop: The System}.]
For all $t\leq 0$ and any choice of $y_{Q_k}$, $k=1,\ldots, N$, take $\bold v_p(t)=\left(\sqrt{N+1},0\right)$, $P(t)=t\bold v_P$, $Q_k(t)=\left(\dfrac kN,y_{Q_k}\right)$, and $\bold v_{Q_k}(t)=(0,0)$. For all $\sigma_N<\dfrac1N$, it is clear that $P,Q_1,\ldots, Q_N$ solve the Hamiltonian system for $t\leq 0$ (as there is no interaction). We now specify $y_{Q_k}$'s for the evolution when $t>0$.

Applying Lemma \ref{successive collision}
for $x_0 = y_0=0$, $r = 0$, $\phi = 0$,  $v  = \sqrt{N+1}$, $d = 1/N$ and $\theta=\theta_1 = \phi_1 =-\arcsin(1/\sqrt{N+1})$ there is $y_{Q_1}$ such that $P$ will interact with $Q_1$ and after interaction
\begin{equation}
\begin{split}
     \bold v_P 
     = 
     \sqrt{N}\left(\cos{\phi_1},\sin{\phi_1}\right),\quad
     \bold v_{Q_1} 
     =
     \left(-\sin{\phi_1},\cos{\phi_1}\right). 
\end{split}
\end{equation}
In this way, the position of $Q_1$, depending on $\sigma_N$, is now determined.
The whole interaction, according to Lemma \ref{successive collision}, takes place in the disc of radius $r_1 =6 \sigma_N$ and center $(1/N, 0)$. Let $[t_1',t_1'']$ be the time interval of this interaction. Preparing for the next interaction, make a new choice of $\sigma_N$ so that $r_1=6 \sigma_N < 1/N$, and note that everything in this first step still holds for the new choice of $\sigma_N$.

For induction, fix $k \in \N$ and assume that $r_1$, \ldots,$r_k$ satisfy \eqref{estimate of r_k}, and therefore $r_j < 1/N$, $j=1,\ldots,k$, for all $\sigma_N$ small enough.
Further assume that $y_{Q_1}$,\ldots,$y_{Q_k}$, $t_1'$,\ldots,$t_k'$,
$t_1''$,\ldots,$t_k''$,
$\bold v_{Q_1}(t)$,\ldots, $\bold v_{Q_k}(t)$,
$\bold v_P(t)$, for $t\leq t_k''$, have all been determined and satisfy \eqref{eq: velocity of P} and \eqref{v_{Q_k}}.

Apply 
Lemma \ref{successive collision} for $(x_0,y_0) = (x_k,y_k)$, for $(x_k,y_k)$ as in \eqref{eq:x_k,y_k}, $r = r_k$, $\phi = \phi_k$, $v = \sqrt{N+1-k}$,  $d = 1/N$ and $\theta=\theta_{k+1}$ as in \eqref{eq:theta&phi}, 
to find that $r_{k+1}$ is determined by formula \eqref{estimate of r_k}, to determine $y_{Q_{k+1}}$, the times $t'_{k+1}$, $t''_{k+1}$, and the velocities $\bold v_P(t)$, $\bold v_{Q_{k+1}}(t)$ for $t\in[t'_{k+1}, t''_{k+1}]$ that will satisfy  \eqref{eq: velocity of P} and \eqref{v_{Q_k}}.  
Therefore $Q_{k+1}$ is always in the $r_{k+1}$-neighborhood of $\mathcal{Q}_{k+1}$, as defined in \eqref{eq:mathcalQ_n}.
Choose $\sigma_N$ so that $r_{k+1}$ is smaller than $1/N$. 
Using Lemma \ref{P_n Q_n}, 
$Q_{k+1}$ does not interact with $Q_1$,\ldots,$Q_k$ during the interval $(-\infty, t_{k+1}'']$. 
\vfill
\begin{figure}
\begin{center}
\begin{tikzpicture}
\draw [dotted] (-6,-2.5) -- (-6,1);
\draw [dotted] (-3,-2.5) -- (-3,1);
\draw [dotted] (0,-2.5) -- (0,1);
\draw [dotted] (3,-2.5) -- (3,1);
\draw (-3,0) circle [radius=0.1];
\draw (0,-1.7321) circle [radius=0.2];
\draw (3,-1.4556) circle [radius=0.6];
\draw [solid] (3,-1.4556) -- (2.75,-0.9264);
\draw [dotted] (-6,0) -- (-3,0) -- (0,-1.7321) -- (3,-1.4556);
\coordinate [label=right:$r_3$]  (r3) at (2.85,-1.2);
\node at (-6,0) {.};
\draw [solid] (-6,0) -- (-3.1,0);
\coordinate [label=below:${P}$]  (P) at (-4.5,0);
\draw [->] (-6,0) -- (-4,0);
\draw [->] (-2.9500, 0.0866) -- (-2.4500, 0.9526);
\coordinate [label=right:${Q_1}$]  (Q1) at (-2.4500, 0.9526);
\draw [->] (-2.9500, -0.0866) -- (-1.4500, -0.9526);
\coordinate [label=below:${P}$]  (P) at  (-1.4500, -0.9526);
\draw [solid]  (-1.4500, -0.9526) -- (-0.1926,-1.6781);
\draw [->] (0.06,-1.92284) -- (0.1518, -2.9187);
\coordinate [label=right:$Q_2$]  (Q2) at   (0.1518, -2.9187);
\draw [->]  (0.15,-1.86433) -- (1.5582 ,  -1.7346);
\coordinate [label=below:${P}$]  (P) at (1.5582 ,  -1.7346);
\draw [solid] (1.5582  , -1.7346) -- (2.4337 , -1.6540);
\draw [->]  (3.4305, -1.0377) -- (4.0697  , -0.2687);
\coordinate [label=right:$Q_3$]  (Q3) at  (4.0697  , -0.2687);
\draw [->] (3.4084,-1.8951) -- (4.1774 ,  -2.5343);
\coordinate [label=below:$P$]  (P) at  (4.1774 ,  -2.5343);
\end{tikzpicture}
\end{center}
\caption{Schematics of the system of Proposition \ref{Prop: The System} after three interactions. The radii discs are not up to scale.}
\label{Fig: three circles}
\end{figure}
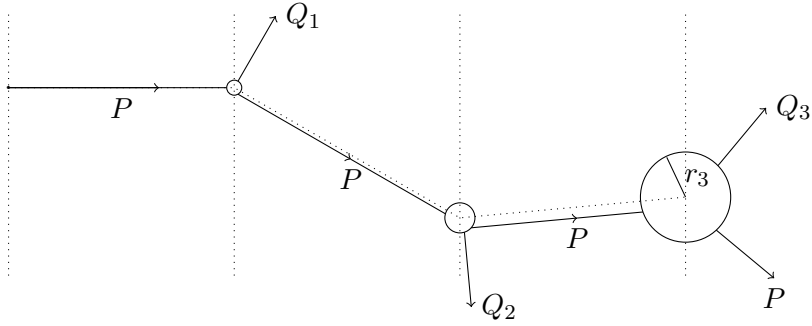

For \eqref{estimate of r_k}, rewrite first \eqref{r_k r_{k-1}} as
\begin{equation}
\begin{split}
         r_k&=\sec \phi_{k-1} r_{k-1} +\sec\phi_{k-1}\sigma_N+5\sigma_N\\
               &=\prod_{j=0}^{k-1} \sec\phi_j r_0 +\sum_{j=0}^{k-1}\prod_{m=j}^{k-1}\sec\phi_m \sigma_N
               +\left(\sum_{j=1}^{k-1}\prod_{m=j}^{k-1}\sec\phi_m+1\right)5\sigma_N\\
\end{split}
\end{equation}
and, using $r_0=0$ and $\left| \phi_j\right|\leq \left| \theta_j \right|$ (Lemma \ref{theta&phi}), estimate this by
\begin{equation}
\begin{split}
               &\leq \sum_{j=0}^{k-1}\prod_{m=j}^{k-1}\sec\theta_m \sigma_N
               +\left(\sum_{j=1}^{k-1}\prod_{m=j}^{k-1}\sec\theta_m+1\right)5\sigma_N, \end{split}
\end{equation}
and then, increasing $k$ to $N$ and using \eqref{eq:theta&phi}, estimate the same by
\begin{equation}
\begin{split}              
               &\leq \sum_{j=0}^{N-1}\dfrac{\sqrt{N+2-j}}{\sqrt 2}\sigma_N
               +\left(\sum_{j=1}^{N-1}\dfrac{\sqrt{N+2-j}}{\sqrt 2}+1\right)5\sigma_N\\
               &\leq 6 \sum_{j=0}^{N-1}\dfrac{\sqrt{N+2-j}}{\sqrt 2}\sigma_N
                 \leq 3\sqrt 2 \sum_{j=3}^{N+2}\sqrt{j}\sigma_N\\
               &\leq 3\sqrt 2 \sigma_N \int_3 ^{N+3} \sqrt x dx<2\sqrt2(N+3)^{3/2}\sigma_N.
\end{split}
\end{equation}
In particular,  $\displaystyle \sigma_N < \frac{1}{2\sqrt{2}}\frac{1}{N(N+3)^{3/2}}$ implies $r_k <1/N$ for all $k$. 
\end{proof}

\begin{Rmk}
Notice that, for each $N$, Proposition \ref{Prop: The System} provides examples of the general theory of Gal'perin and Vaserstein, \cite{G} and \cite{V}, according to which, for finite range interactions, molecules evolve by eventually separating into independent clusters. Each cluster here consists of a single molecule.
\end{Rmk}

\subsection{The limit system as $N \to \infty$} \label{subsection: N to infty}
In the notation of Proposition \ref{Prop: The System}, let 
 $T_N' := \displaystyle \sum_{j = 1}^N (t_j'' - t_j')$, the time during which $P$ interacts with some $Q_k$.
Then $T_N''=t_N'' - T_N' $ is the time during $[0,t_N'']$ when $P$ is not interacting at all.

\begin{Prop} \label{TN to 0} 
$t_N'' \to 0$, as $N \to \infty$.
\end{Prop}
\begin{proof}
According to \eqref{eq: velocity of P}, the speed of $P$ at $t_k'$ is $\displaystyle \sqrt{N+2-k}$. Then, by Lemma \ref{Lem: CollisionTimeEstimate},
\begin{equation}\label{T_N'}
   \begin{split}
        T_N'
        <
        \sum_{k=1}^N
        \dfrac{4\sigma_N}{\sqrt{N+2-k}}
        =
        4\sigma_N\sum_{k=2}^{N+1} \dfrac{1}{\sqrt{k}}.
   \end{split}
\end{equation}
After the interaction of $P$ with $Q_k$ is complete, $P$ moves with speed $\sqrt{N+1-k}$, forming angle $\phi_k$ with the $x$-axis. The distance $d_k$ that $P$ will travel until its interaction with $Q_{k+1}$ begins, satisfies
 \begin{equation}
   \begin{split}
        d_k
        \leq
        \dfrac {1}{N\cos\phi_k}
        \leq
        \dfrac {1}{N\cos\theta_k},
   \end{split}
\end{equation}
cf.\ Figure \ref{fig: the single step}.
Recalling that $|\theta_k|\leq\dfrac{\pi}{4}$ from Lemma \ref{theta&phi} gives 
 \begin{equation}
   \begin{split}
        T_N''
        &<
        \sum_{k=0}^{N-1} 
        \dfrac {1}{N\cos\theta_k}
        \dfrac{1}{\sqrt{N+1-k}}\\
        &\leq
        \dfrac{\sqrt{2}}{N}
        \sum_{k=0}^{N-1} 
        \dfrac{1}{\sqrt{N+1-k}}
        =
        \dfrac{\sqrt{2}}{N}
        \sum_{k=2}^{N+1} 
        \dfrac{1}{\sqrt{k}}.
   \end{split}
\end{equation}
This and \eqref{T_N'} imply
 \begin{equation}
   \begin{split}
        t_N''
        <
        \left(4\sigma_N+\dfrac {\sqrt{2}}N\right)\sum_{k=2}^{N+1} \dfrac{1}{\sqrt{k}}. 
   \end{split}
\end{equation}
As $\displaystyle \sum_{k=2}^{N+1} \dfrac{1}{\sqrt{k}} < 2\sqrt{N+1}$, and for $\sigma_N$ as in Proposition \ref{Prop: The System}, we conclude that $t''_N\to0$ as $N\to \infty$.
\end{proof}

\begin{Prop} \label{yQk to 0}
$\displaystyle \max_{0\leq k\leq N}y_{Q_k}\to 0$ as $N\to\infty$.
\end{Prop}
\begin{proof}
Noting that $y_{Q_k}$ is the second coordinate of $Q_k$ before $t=t_k'$, whereas $y_k$ is the second coordinate of the center of the $k$-interaction disc, it follows from the definition of $r_k$ and \eqref{estimate of r_k} that 
\begin{equation}
    \left|y_{Q_k}\right|
    <
    \left|y_k\right|
    +
    2\sqrt 2 (N+3)^{3/2}\sigma_N.
\end{equation}
For the second term on the right use $\sigma_N$ as in Proposition \ref{Prop: The System} and estimate the first term as
 \begin{equation}
   \begin{split}
        \left|y_k\right|
        &=
        \dfrac1N \left|\sum_{m=0}^{k-1} \tan\phi_m\right|
        \leq
        \dfrac 1N \sum_{m=0}^{k-1} \tan \left|\phi_m\right| \\
        &\leq
        \dfrac 1N \sum_{m=0}^{k-1} \tan \left|\theta_m\right|
        =
        \dfrac 1N\sum_{m=0}^{k-1} \tan \left(\arcsin\dfrac{1}{ \sqrt{N+2-m}}\right)\\
        &=
        \sum_{m=0}^{k-1}\dfrac{1}{N\sqrt{N+1-m}}
        <
         \sum_{m=0}^{N-1}\dfrac{1}{N\sqrt{N+1-m}}\\
         &=
         \dfrac 1N\sum_{m=2}^{N+1}\dfrac{1}{\sqrt m}
         \to 0,
   \end{split}
\end{equation}
as $N \to \infty$. 
\end{proof}
For each fixed $N$, writing ${\bold v} = (v_x, v_y)$ and following \eqref{M_t^{N}}, set for $t\in\R$
\begin{equation} \label{Main Theorem}
   \begin{split}
        &M_t^{(N+1)}(dx,dy,dv_x,dv_y)\\
        &\ \ \ \ \ =
        \dfrac 1{N+1}
        \left(
        \delta_{\left(P(t),\bold v_P(t)\right)}(dx,dy,dv_x,dv_y)
        +
        \sum_{k=1}^N \delta_{\left(Q_k(t),\bold v_{Q_k}(t)\right)}(dx,dy,dv_x,dv_y)
        \right).
    \end{split}
\end{equation}
The crucial observation in the following proposition is that, due to the factor $1/N$, no single molecule shows as $N\to 
\infty$, but its interaction with many other molecules, if their number is of order $N$, shows macroscopically.
\begin{Prop} \label{Thm: Main Theorem}
As $N\to \infty$, and  for $\sigma_N$ as in Proposition \ref{Prop: The System}: for $t\leq 0$,
\begin{equation} \label{M_t for T negative}
   \begin{split}
       M_t^{(N+1)}(dx,dy,dv_x,dv_y)
        \Rightarrow
        \chi_{[0,1]}(x)dx
        \otimes
        \delta_0(dy)
        \otimes 
        \delta_{(0,0)}(dv_x,dv_y),
    \end{split}
\end{equation}
and for $t>0$,
\begin{equation} \label{M_t for T positive}
   \begin{split}
        &M_t^{(N+1)}(dx,dy,dv_x,dv_y)\\
        &\ \ \ \ \Rightarrow 
        \chi_{[0,1]}(x)dx
        \otimes
        \left(
        \frac12\delta_t(dy)\otimes  \delta_{(0,1)}(dv_x,dv_y) 
        +
        \frac12 \delta_{-t}(dy) \otimes\delta_{(0,-1)}(dv_x,dv_y)  
    \right).
    \end{split}
\end{equation} 
\end{Prop}
\begin{proof}
It suffices to check the statement on the integrals of bounded Lipschitz functions, see \cite{AGS}, page 109.
For this, for  $f:\mathbb R^2\times \mathbb R^2 \to \R$ bounded and Lipschitz
\begin{equation}
   \begin{split}
        &\int_{\R^4} f \left(x,y,v_x,v_y\right)M_t^{(N+1)}(dx,dy,dv_x,dv_y)\\
        &\ \ \ \ \ \ =
          \dfrac 1 {N+1} f\left(P(t),\bold v_P(t)\right)+
         \dfrac 1 {N+1} \sum_{k=1}^{N} f \left(Q_k(t),\bold v_{Q_k}(t)\right).
    \end{split}
\end{equation}
Since $f$ is bounded,  the first term vanishes as $N\to \infty$. The rest of the proof examines the convergence of the second term.

Fix any $t\leq0$. Recalling \eqref{P_k initial setting},
\begin{equation}
   \begin{split}
        \dfrac 1 {N+1} \sum_{k=1}^{N} f \left(Q_k(t),\bold v_{Q_k}(t)\right)
        =
         \dfrac 1 {N+1} \sum_{k=1}^{N} f \left(\dfrac kN,y_{Q_k},0,0\right).
    \end{split}
\end{equation}
For $L_f$ be the Lipschitz constant of $f$, 
and using Proposition \ref{yQk to 0},
\begin{equation}
   \begin{split}
        &\left| 
        \dfrac 1 {N+1} \sum_{k=1}^{N}f\left(\dfrac kN,y_{Q_k},0,0\right)
        -
         \dfrac 1 {N+1} \sum_{k=1}^{N}f\left(\dfrac kN,0,0,0\right)
        \right|\\
        &\ \ \ \ \ \ \ \ \ \ \ \ \ \ \ \leq
        \dfrac {N}{N+1} L_{f} \max_{1\leq k\leq N}\left|y_{Q_k}\right|\to 0.
    \end{split}
\end{equation}
By the definition of the Riemann integral, 
\begin{equation}
   \begin{split}
         \dfrac 1 {N+1} \sum_{k=1}^{N}f\left(\dfrac kN,0,0,0\right)
          \to
          \int_0^1 f\left( x,0,0,0\right)dx.
    \end{split}
\end{equation}
Therefore
\begin{equation}
   \begin{split}
        \int_{\R^4} f\left(x,y,v_x,v_y\right)M_t^{(N+1)}(dx,dy,dv_x,dv_y)
        \to
          \int_0^1 f\left( x,0,0,0\right)dx.
    \end{split}
\end{equation}
This is exactly \eqref{M_t for T negative}.
Now fix $t>0$.
By Proposition \ref{TN to 0} there exists $N_1$ such that for all $N>N_1$, $t_N''<t$, i.e.\ for each time we can choose $N$ large enough so that all interactions have already happened and all molecules {\em are} moving at time $t$, and are moving with their terminal velocities. We consider such $N$'s only.
According to  Proposition \ref{Prop: The System}, and since now $t \geq t_k''$, 
\begin{equation}
   \begin{split}
          x_{Q_k}(t)
          &=
          x_{Q_k}(t_k'')+\left(t- t_k''\right)v_{Q_k,x}(t_k''),\\
          y_{Q_k}(t)
          &=
          y_{Q_k}(t_k'')+\left(t- t_k''\right)v_{Q_k,y}(t_k'').
   \end{split}
\end{equation}
For $\alpha _N=\left[ N-\sqrt N\right]$, the integer part of $N-\sqrt N$, and by \eqref{v_{Q_k}},  for any $1\leq k\leq \alpha_N$
\begin{equation} \label{EQ: Estimates on t_k''}
   \begin{split}
    &\left| v_{Q_k,x}(t_k'')\right| =\sin | \phi_k |  \leq  \sin |\theta_k | \leq \dfrac{1}{\sqrt{N+2-\alpha_N}},\\
    &\left| v_{Q_k,y}(t_k'') - (-1)^{k+1} \right| =| \cos \phi_k -1| \leq | \sin {\phi_k}| \leq \dfrac{1}{\sqrt{N+2-\alpha_N}},\\
    &\left|x_{Q_k}(t_k'')-\dfrac kN\right|\leq r_k,
    \quad
    \left|y_{Q_k}(t_k'')\right|<\left|y_k\right|+r_k.
   \end{split}
\end{equation}
Therefore for $1\leq k\leq \alpha_N$,
by \eqref{estimate of r_k}, Proposition \ref{TN to 0}, and Proposition \ref{yQk to 0},
\begin{equation} \label{Eq: Q_k estimate}
   \begin{split}
    \left| x_{Q_k}(t)-\dfrac k N\right|
    & \leq \left|x_{Q_k}(t_k'')-\dfrac kN\right| + \left(t- t_k''\right) \left| v_{x,Q_k}(t_k'')\right|\\
    & < r_k+\dfrac {t}{\sqrt{N+2-\alpha_N}}\to 0,\\
    \left| y_{Q_k}(t)-(-1)^{k+1} t\right|
    & \leq \left|y_{Q_k}(t_k'')\right| +t \left|v_{Q_k,y}(t_k'')-(-1)^{k+1}\right|+t_k''\left| v_{Q_k,y}(t_k'') \right|\\
    & < \left|y_k\right|+r_k+\dfrac {t}{\sqrt{N+2-\alpha_N}} + t_k''\to 0.
   \end{split}
\end{equation}
Since $f$ is Lipschitz,  \eqref{EQ: Estimates on t_k''} and \eqref{Eq: Q_k estimate} imply that  
\begin{equation}
   \begin{split}
        \left| 
       \dfrac 1 {N+1} \sum_{k=1}^{\alpha_N}f\left(Q_k(t),\bold v_{Q_k}(t)\right)
        -
         \dfrac 1 {N+1} \sum_{k=1}^{\alpha_N}f\left(\dfrac k N,(-1)^{k+1} t, 0,(-1)^{k+1}\right)
         \right|
         \to 0.
    \end{split}
\end{equation}
For $C_{f}=\max \left|f\right|$, 
\begin{equation}
   \begin{split}
         &\left|
         \dfrac 1 {N+1} \sum_{k=\alpha_N+1}^{N}f\left(Q_k(t),\bold v_{Q_k}(t)\right)
         \right|
         \leq C_{f}\dfrac{N-\alpha_N}{N+1}
         \to 0,\\
         &\left|
         \dfrac 1 {N+1} \sum_{k=\alpha_N+1}^{N}f\left(\dfrac kN,(-1)^{k+1} t, 0,(-1)^{k+1}\right)
         \right|
         \leq C_{f}\dfrac{N-\alpha_N}{N+1}
         \to 0,
    \end{split}
\end{equation}
therefore,
\begin{equation}
   \begin{split}
        \left| 
       \dfrac 1 {N+1} \sum_{k=1}^{ N}f\left(Q_k(t),\bold v_{Q_k}(t)\right)
        -
         \dfrac 1 {N+1} \sum_{k=1}^{N}f\left(\dfrac kN,(-1)^{k+1} t, 0,(-1)^{k+1}\right)
         \right|
         \to 0.
    \end{split}
\end{equation}
By the definition of the Riemann integral, 
\begin{equation}
   \begin{split}
        \dfrac 1 {N+1} \sum_{k=1}^{N}f\left(\dfrac kN,(-1)^{k+1} t, 0,(-1)^{k+1}\right)
        \to&
        \int_0^1 \dfrac12\left(f(x,t,0,1)+ f(x,-t,0,-1)\right)dx
    \end{split}
\end{equation}
which implies \eqref{M_t for T positive}.
\end{proof}
With
\begin{equation}
\begin{split}
&\bold x_k^{(N+1)}(t)= Q_k(t), \quad
\bold u_k^{(N+1)}(t)= \bold v_{Q_k}(t), \  k=1,\ldots,N\\
&\bold  x_{N+1}^{(N+1)}(t) = P(t),  \quad
\bold  u_{N+1}^{(N+1)}(t) = \bold v_{P}(t),
\end{split}
\end{equation}
 Theorem \ref{Thm: ghost example} follows immediately from Propositions \ref{Prop: The System} and \ref{Thm: Main Theorem}.

\subsection{Macroscopic equations} \label{Macroscopic Equations}
We now examine the hydrodynamic equations for $M_t(d\bold x,d\bold v)$ as in Theorem \ref{Thm: ghost example}.
 It is easy to check that for any $\phi(t,\bold x)\in C_c^{\infty} (\R \times \R^2)$
\begin{equation}\label {eq: continuity momentum in ghost }
\begin{split}
   &
   \int_{-\infty}^{\infty} \int_{\R^4} 
   \partial_t\phi(t,\bold x) M_t(d\bold x,d\bold v)
   dt
    +
    \int_{-\infty}^{\infty}  \int_{\R^4}  
    \nabla_{\bold x}\phi(t,\bold x) \cdot \bold v 
    M_t(d\bold x,d\bold v)dt
    =
    0,\\
    &
   \int_{-\infty}^{\infty}  \int_{\R^4} 
    \partial_t\phi(t,\bold x) \bold v M_t(d\bold x,d\bold v)dt
    +
    \int_{-\infty}^{\infty}  \int_{\R^4} 
    \nabla_{\bold x}\phi(t,\bold x) \cdot \bold v \, \bold v    
    M_t(d\bold x,d\bold v)dt
    =
    0.
\end{split}
\end{equation}
Using disintegration \eqref{eq: disintegration of ghost Mt}, for $\mu_t(d\bold x)$ and $\bold u(t,\bold x)$ as in  
\eqref{eq: macro density from ghost} and \eqref{eq: macro velocity from ghost},
 we rewrite \eqref{eq: continuity momentum in ghost } as 
\begin{equation}\label{eq: rewritten form continuity momentum in ghost }
\begin{split}
   &
   \int_{-\infty}^{\infty} \int_{\R^2} 
   \partial_t\phi(t,\bold x) 
   \mu_t(d\bold x) dt
    +
    \int_{-\infty}^{\infty}  \int_{\R^2}  
    \nabla_{\bold x}\phi(t,\bold x) \cdot 
    \bold u(t,\bold x)
     \mu_t(d\bold x) dt
    =
    0,\\
    &
   \int_{-\infty}^{\infty}  \int_{\R^2} 
    \partial_t\phi(t,\bold x)
     \bold u(t, \bold x)
     \mu_t(d\bold x) dt\\
    &\quad\quad\quad\quad\ \  + 
    \int_{-\infty}^{\infty}  \int_{\R^2} 
    \nabla_{\bold x}\phi(t,\bold x) \cdot
    \left(  \int_{\R^2}
    \bold v \otimes \bold v
    M_{t,\bold x}( d\bold v)
    \right)
    \mu_t(d\bold x)dt
    =
    0.
\end{split}
\end{equation}
Notice that at each $t, \bold x$  the $M_{t,\bold x}(d \bold v)$ is singular, therefore
\begin{equation}
      \int_{\R^2}
    \bold v \otimes \bold v
    M_{t,\bold x}( d\bold v)
    =
    \bold u \otimes \bold u.
\end{equation}
Then \eqref{eq: rewritten form continuity momentum in ghost } becomes
\begin{equation} \label {system: weak p-less euler}
\begin{split}
   &
   \int_{-\infty}^{\infty} \int_{\R^2} 
   \partial_t\phi(t,\bold x) \mu_t(d\bold x)
   dt
    +
    \int_{-\infty}^{\infty}  \int_{\R^2}  
    \nabla_{\bold x}\phi(t,\bold x) \cdot \bold u 
    \, \mu_t(d\bold x) dt
    =
    0,\\
    &
   \int_{-\infty}^{\infty}  \int_{\R^2} 
    \partial_t\phi(t,\bold x) \bold u \, \mu_t(d\bold x) dt
    +
    \int_{-\infty}^{\infty}  \int_{\R^2} 
    \nabla_{\bold x}\phi(t,\bold x) \cdot \bold u \, \bold u   
     \,  \mu_t(d\bold x)dt
    =
    0.
\end{split}
\end{equation}
In other words $\left(\mu_t(d\bold x), \bold u(t,\bold x)\right)$, $t\in \R$  solves weakly two dimensional Euler system without pressure:
\begin{equation} \label {system: p-less euler}
\begin{split}
     &\partial_t \mu_t+ \div(\bold u\, \mu_t ) 
     = 0,\\
     &\partial_t( \bold u\, \mu_t) + \div( \bold u\otimes \bold u\ \mu_t) 
     =0.
\end{split}
\end{equation}
{For the naturalness of measure solutions in the presureless Euler system, see \cite{ERS}, p.\,354.}
\begin{Rmk}
The trivial solution $\widetilde \mu_t(d\bold x)=  \Delta_{0}(d\bold x) , \widetilde{ \bold u}=(0,0)$ also solves \eqref{system: weak p-less euler} for all $t$, and coincides with $\left( \mu_t, \bold u \right)$  for $t\leq 0$. Note that $\left( \mu_t, \bold u \right)$ is not ``energy admissible" since the kinetic energy of $\left( \mu_t, \bold u \right)$ increases in time:
\begin{equation}
\int_{\R^2} | \bold u|^2 \mu_t(d \bold x)
=
\int_{\R^4} |\bold v|^2 M_t(d\bold x,d\bold v)
=
\left \{
\begin{array}{ll}
0 & t\leq 0  \\
1 & t>0 .
\end{array}
\right.
\end{equation}
A solution to \eqref{system: weak p-less euler} with decreasing energy can be obtained by reversing the direction of time, as in the next section. The value of the construction in this section lies in the microscopic, Hamiltonian interpretation of spontaneous velocity generation in weak solutions of hydrodynamic equations as in \cite{Sch}, \cite{Sh}.  
\end{Rmk}


 \section{Time Reversal and Macroscopic Non-Uniqueness}\label{section:reversal}
\subsection{Reverse flow with decreasing energy} 
\label{subsection reverse flow}
 We now reverse time in the construction of the previous section to establish macroscopic non-uniqueness in the class of energy decreasing solutions.
It is standard that for $\left( \bold x_k^{(N)}(t), \bold u_k^{(N)}(t)\right)$ a Hamiltonian flow  as in Theorem \ref{Thm: ghost example} the reverse flow $\left( \bold x_k^{(N)}(-t), -\bold u_k^{(N)}(-t)\right)$ also solves the Hamiltonian system \eqref{Ham}.
{Roughly speaking, for each $N$ the reverse system consists of $N$ molecules moving with speed $1$ for $t<0$. At $t=0$, through interaction, one of the $N$ molecules gathers all the energy from the rest $N-1$ molecules and leaves the rest of the group. Therefore for $t>0$, macroscopically the system is motionless. If we still use $\left( \bold x_k^{(N)}(t), \bold u_k^{(N)}(t)\right)$ for the reverse flow then the measure $M_t^{(N)}$ converges weakly to
} 
\begin{equation}\label{eq: M_t reverse ghost for all t}
M_t(d\bold x, d \bold v)
=
\left\{
\begin{array}{ll}
         \dfrac12
         \Delta_{t}(d\bold x)
         \otimes
         \delta_{\left(0,1\right)}\left(d\bold v\right)
         +
         \dfrac 12
         \Delta_{-t}(d\bold x)
         \otimes
         \delta_{\left(0,-1\right)}\left(d\bold v\right)
      &
        t < 0 
        \\
         \Delta_0(d\bold x)
         \otimes
         \delta_{\left(0,0\right)}\left(d\bold v\right)
      &
      t\geq 0,
\end{array}
\right.
\end{equation} 
%
with
\begin{equation} \label {eq: mu_t u_t from t-ghost}
\begin{split}
     \mu_t(d\bold x) 
     &=
     \left\{
     \begin{array}{ll}
     \dfrac12 \Delta_t(d\bold x)  + \dfrac 12 \Delta_{-t}(d\bold x)
     & t< 0
     \\
     \Delta_0(d\bold x) 
     & t\geq 0,
     \end{array}
     \right. \\
     \bold u(t,\bold x) 
     &=
     \left\{
     \begin{array}{ll}
    \chi_{\mathbf Q_t}(\bold x) \cdot (0,1) +  \chi_{\mathbf Q_{-t}}(\bold x) \cdot (0,-1)
      & t< 0\\
     0 &  t\geq 0,
     \end{array}
     \right.
\end{split}
\end{equation} 
and decreasing energy:
\begin{equation} \label{eq: energy reverse flow}
\int_{\R^2} | \bold u|^2 \mu_t(d \bold x)
=
\int_{\R^4} |\bold v|^2 M_t(d\bold x,d\bold v)
=
\left \{
\begin{array}{ll}
1 & t<0\\
0 & t\geq 0,
\end{array}
\right.
\end{equation}
{cf. \cite{BN}, Defintion 2.1.}
\begin{Rmk}
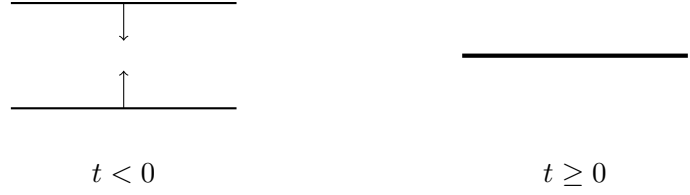
\begin{figure}
\begin{tikzpicture}
  \coordinate []  (A1) at (0,0.7);
  \coordinate [] (A2) at (3,0.7);
  \coordinate []  (B1) at (0,-0.7);
  \coordinate [] (B2) at (3,-0.7);
  \coordinate [] (C1) at (6,0);
  \coordinate [] (C2) at (9,0);
    \coordinate [] (M1) at (1.5,0.7);
  \coordinate [] (M2) at (1.5,0.2);
    \coordinate [] (N1) at (1.5,-0.7);
  \coordinate [] (N2) at (1.5,-0.2);
  \draw [thick](A1) -- (A2);
  \draw [thick](B1) -- (B2);
  \draw [ultra thick](C1) -- (C2);
  \draw [->](M1) -- (M2);
  \draw [->](N1) -- (N2);
  \coordinate [label=below:$t < 0$] (v) at (1.5,-1.3);
  \coordinate [label=below:$t \geq 0$] (v) at (7.5,-1.3);
   \end{tikzpicture}
    \caption{Macroscopic flow of \eqref{eq: M_t reverse ghost for all t}.}
    \label{Figure: Reverse Flow}
\end{figure}
This describes two fronts approaching each other up until $t=0$, when they merge and stay at rest, see Figure \ref{Figure: Reverse Flow}. In the context of the pressureless Euler system this is a ``sticky" macroscopic solution, cf. \cite{BN}. Rather than using particle systems with adhesion dynamics, here we obtain the solution as the limit of Hamiltonian dynamics with repulsive force. We also provide an explanation for the loss of energy:  all the energy is transferred to a macroscopically invisible part of the system.
\end{Rmk}

\subsection{Transverse flow}
\label{reversals}
It is known that merely requiring decreasing energy does not guarantee uniqueness of measure solutions to the system \eqref{system: p-less euler}, see \cite{BN}. This persists when comparing the flow of the previous section with the limit of a trivial Hamiltonian flow: for this we take the $N$-system to consist of molecules that stay far enough from each other so that they never interact. We obtain a solution to the system \eqref{system: weak p-less euler} that coincides with \eqref{eq: mu_t u_t from t-ghost} for all $t < 0$.  But at $t=0$, the moment the two fronts meet, instead of merging and staying at rest, they go through each other. 

More precisely, for each $N=2n \in \N$, $j=1,2,\ldots,N$, let
\begin{equation} \label{eq: square initial conditions}
\begin{split}
       \widetilde {\bold x}_{j}^{(N)}
        =
        \left(
        \dfrac{j}{N},
        0
        \right),
        \quad
       \widetilde {\bold u}_j^{(N)}
       =
       \begin{cases}
        (0,1)&\mbox{ if $j$ odd}\\
        (0,-1)&\mbox{ if $j$ even} .
        \end{cases}
    \end{split}
\end{equation}
For $t\in \R$ the orbits  
\begin{equation}\label{eq: free flow slns}
\begin{split}
    \widetilde {\bold x}_{j}^{(N)}(t) = \widetilde {\bold x}_{j}^{(N)} + t \widetilde {\bold u}_{j}^{(N)}
\end{split}
\end{equation} 
satisfy the Hamiltonian system \eqref{Ham}
provided that the interaction range is sufficiently short, for example, $\sigma<1/N$. (Notice that $\sigma_N$ in Theorem \ref{Thm: ghost example}, and therefore in Section \ref{section:reversal} satisfies $\sigma_N<1/N$.)
Recalling definition \eqref{M_t^{N}}, set
\begin{equation}
   \begin{split}
         \widetilde M_t^{(N)}(d\bold x,d\bold v)
         =
         \dfrac{1}{N} \sum_{j=1}^{N} \delta_{\left(\widetilde {\bold x}_{j}^{(N)}(t), \widetilde {\bold u}_{j}^{(N)}(t)\right)}
         (d\bold x,d\bold v).
   \end{split}
\end{equation} 
By the definition of Riemann integral, for any continuous bounded $f(\bold x,\bold v)$ we have 
\begin{equation}   \label{eq: limit of square at t}
   \begin{split}
        \lim_{N\to\infty}
         \int_{\R^4} f(\bold x,\bold v) \widetilde M_t^{(N)}(d\bold x,d\bold v)
         =
         \frac12\int_0^1
         f\left(x,t,0,1\right)
         dx\
         +
         \frac12\int_0^1
         f\left(x,-t,0,-1\right)
         dx.
   \end{split}
\end{equation}
Therefore 
\begin{equation} \label{eq: M_t in square}
\begin{split}
       \widetilde M_t^{(N)}(d\bold x,d\bold v)
        & \Rightarrow
         \widetilde M_t(d\bold x,d\bold v)\\
       : & =
       \dfrac12
         \Delta_{t}(d\bold x)
         \otimes
         \delta_{\left(0,1\right)}\left(d\bold v\right)
         +
         \dfrac 12
         \Delta_{-t}(d\bold x)
         \otimes
         \delta_{\left(0,-1\right)}\left(d\bold v\right)
         .      
\end{split}
\end{equation}
The macroscopic density and velocity are
\begin{equation} \label{eq: mu_t u_t p-less from planes }
\begin{split}
     \widetilde \mu_t (d\bold x)
      &= 
      \dfrac12
         \Delta_{t}(d\bold x)
         +
      \dfrac 12
         \Delta_{-t}(d\bold x),\\
       \widetilde {\bold u}(t,\bold x)
        &=
       \chi_{\mathbf Q_t}(\bold x) \cdot (0,1) +  \chi_{\mathbf Q_{-t}}(\bold x) \cdot (0,-1),
       \quad t\in \R,
   \end{split}
\end{equation}
see Figure \ref {Transverse Flow}.
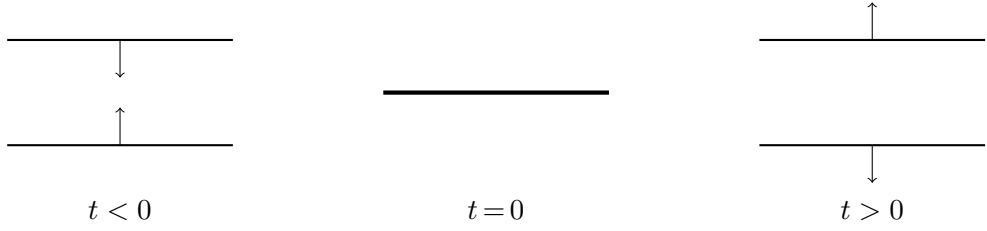
\begin{figure}
\begin{tikzpicture}
  \coordinate []  (A1) at (0,0.7);
  \coordinate [] (A2) at (3,0.7);
  \coordinate []  (D1) at (0,-0.7);
  \coordinate [] (D2) at (3,-0.7);
   \coordinate [] (E1) at (5,0);
  \coordinate [] (E2) at (8,0);
  \coordinate [] (B1) at (10,0.7);
  \coordinate [] (B2) at (13,0.7);
  \coordinate [] (C1) at (10,-0.7);
  \coordinate [] (C2) at (13,-0.7);
    \coordinate [] (M1) at (11.5,0.7);
  \coordinate [] (M2) at (11.5,1.2);
    \coordinate [] (N1) at (11.5,-0.7);
  \coordinate [] (N2) at (11.5,-1.2);
  \coordinate [] (P1) at (1.5,0.2);
  \coordinate [] (P2) at (1.5,0.7);
  \coordinate [] (Q1) at (1.5,-0.2);
  \coordinate [] (Q2) at (1.5,-0.7);
  \draw [thick](A1) -- (A2);
  \draw [thick](D1) -- (D2);
  \draw [thick](B1) -- (B2);
  \draw [thick](C1) -- (C2);
   \draw [ultra thick](E1) -- (E2);
  \draw [->](M1) -- (M2);
  \draw [->](N1) -- (N2);
  \draw [->](P2) -- (P1);
  \draw [->](Q2) -- (Q1);
  \coordinate [label=below:$t < 0$] (v) at (1.5,-1.3);
  \coordinate [label=below:$t > 0$] (v) at (11.5,-1.3);
    \coordinate [label=below:$t\, {=}\, 0$] (v) at (6.5,-1.3);
   \end{tikzpicture}
    \caption{Macroscopic flow of \eqref {eq: mu_t u_t p-less from planes }. }
    \label{Transverse Flow}
\end{figure}
It is easily checked that \eqref {eq: continuity momentum in ghost }, 
\eqref {eq: rewritten form continuity momentum in ghost } 
hold, and that for all $t\neq 0$
\begin{equation}
      \int_{\R^2}
    \bold v \otimes \bold v
    \widetilde M_{t,\bold x}( d\bold v)
    =
    \widetilde {\bold u} \otimes \widetilde {\bold u}.
\end{equation}
Therefore $\left( \widetilde \mu_t(d\bold x), \widetilde {\bold u}(t,\bold x)\right)$ also solves weakly the pressureless Euler system for $t\in \R$.
Since $\displaystyle \int_{\R^2} |\bold {\widetilde u}|^2 \widetilde \mu_t(d\bold x) =1 $ except for $t=0$,
we can alter $\widetilde {\bold u}$ at time $t=0$ so that
 \begin{equation}
    \int_{\R^2} |\bold {\widetilde u}|^2 \widetilde \mu_0(d\bold x) =1,
\end{equation}
still solving equation \eqref {system: weak p-less euler}.  
If we still use $\widetilde \mu_t(d\bold x), \widetilde {\bold u}(t,\bold x)$ for the modified solution, we then have constant macroscopic kinetic energy in time:
\begin{equation}
\int_{\R^2} | \widetilde {\bold u}|^2 \widetilde \mu_t(d \bold x)
=
1,
\quad t\in \R.
\end{equation}
Clearly for all $t<0$, $(\widetilde \mu_t(d\bold x), \widetilde {\bold u}(t,\bold x))$, modified or not, coincides with $(\mu_t(d\bold x),  {\bold u}(t,\bold x))$.
Macroscopically, the same two fronts are approaching each other and, unless we know their microscopic origin, we are not be able to tell what will happen for $t>0$.

\begin{Rmk}
Notice here the {\em total} macroscopic energy of the limit system is conserved in time:
\begin{equation}\label{eq: total kinetic energy from transverse}
\int_{\R^4} |\bold v|^2 \widetilde M_t(d\bold x, d\bold v) = 1, \quad t\in \R,
\end{equation}
and the macroscopic kinetic energy $\displaystyle \int_{\R^2} |\bold{\widetilde u}|^2 \widetilde\mu_t (d\bold x)$ is only part of the total energy in general:
\begin{equation}
\int_{\R^4} |\bold v|^2 \widetilde M_t(d\bold x, d\bold v)
=
\int_{\R^2} 
|\bold {\widetilde u}|^2 \mu_t(d\bold x) 
+
\int_{\R^4} |\bold v- \bold{\widetilde u} |^2 \widetilde M_t(d\bold x, d\bold v).
\end{equation}
Let $\displaystyle h(t)=\int_{\R^4} \left| \bold v- \widetilde{\bold u}  \right|^2 \widetilde M_t(d\bold x,d\bold v)$.
Then
\begin{equation}
\int_{\R^2} 
|\bold {\widetilde u}|^2 \widetilde \mu_t(d\bold x) + h(t)= 1, \quad t\in \R.
\end{equation}
Notice that $h(t)=0$ when $t\neq 0$ and $h(0)=1$.
Therefore for $t<0$, all the energy of the system \eqref{eq: M_t in square} is macroscopic kinetic energy which becomes $h(0)$, the fluctuation energy, at $t=0$. 
For  $t>0$ all the energy is again macroscopic kinetic energy.

By \eqref{eq: energy reverse flow}, for the reverse flow in Section \ref{subsection reverse flow}, the total energy $\displaystyle \int_{\R^4} |\bold v|^2  M_t(d\bold x, d\bold v)$ is decreasing in time. Trivially, the corresponding fluctuation energy $h(t)= 0$ for all $t\in \R$.
\end{Rmk}

\begin{Rmk}
It is possible that from a Statistical Mechanics point of view the non-uniqueness described here can be avoided by excluding a set of flows $M_t$ negligible with respect to some probability measure. Notwithstanding this, our aim here is to understand specific non-uniqueness examples. 
\end{Rmk}

 
\section{Non-Uniqueness from Moments of Measures Satisfying Identical Transport Equations}
\label{layers}
Section \ref{section:reversal} has shown non-uniqueness by comparing moments of the two limit flows $M_t(d\bold x,d \bold v)$ of \eqref{eq: M_t reverse ghost for all t}  and $\widetilde M_t(d\bold x, d \bold v)$ of \eqref{eq: M_t in square}. Note that $M_t$ satisfies weakly the transport equation
\begin{equation} \label{eq: M-eqn}
\partial_t M_t  + \bold v\cdot \nabla_{\bold x} M_t=0,
\end{equation}
while $\widetilde M_t$ satisfies the same with a nonzero kick at $t=0$:
\begin{equation}
\begin{split}
      \partial_t \widetilde M_t 
      +
      \bold v\cdot \nabla_{\bold x} \widetilde M_t
      = 
      \left(
      \widetilde M_{0^+} - \widetilde M_{0^-}
      \right)
      \otimes
      \delta_0(dt),
      \quad t\in \R,      
\end{split}
\end{equation}
for $\displaystyle \widetilde M_{0^{\pm}} = \lim_{t\to 0^{\pm}} \widetilde M_t$. 
In this section we present two examples where two different measures solve the same transport equation \eqref{eq: M-eqn}, give identical macroscopic density and velocity at $t=0$, but the macroscopic density and velocity evolve differently to provide a non-uniqueness result for the Cauchy problem of the compressible Euler system in space dimension one. 

\subsection{Finite systems with velocity exchange}
\label{sec: Zemlyakov}
For systems in space dimension $1$, we use identical molecules that move freely until they collide. The arguments in this section also hold for systems \eqref{Ham} of (finite range, at least) interactions, rescaled as in \eqref{Uinteraction}. In fact, there exist $\sigma_N$'s such that, for space dimension $1$, the limit of elastic collisions coincides with  the limit of rescaled interactions, see \cite{X}. However, such $\sigma_N$'s might be too small for the rescaled interaction model to be physically better than elastic collisions. 
For simplicity then, we shall use elastic collisions. The complications of finite range interactions were evident in Section \ref{Section:Gost}.

In the elastic collision model collisions are instantaneous. Momentum and energy are conserved.
Here it will be enough to consider only two kinds of collisions, both compatible with finite range interaction dynamics:
\begin{enumerate}
\item
Binary collisions with incoming velocities $v_1$, $v_2$ and outgoing velocities $v_1'$, $v_2'$ satisfying
\begin{equation}
\begin{split}
    \left.
    \begin{array}{l}
    v_1 +v_2 = v_1'+v_2'\\ 
    v_1^2 +v_2^2 = (v_1')^2+(v_2')^2
    \end{array}
    \right\}
    \Rightarrow
    v_1 = v_2', \ v_2 = v_1',
\end{split}
\end{equation} 
i.e.\ the molecules exchange velocities (as they are not allowed to go through each other). 
\item
Triple collisions, consisting of two molecules exactly as in item (1) and a third molecule in between that stays motionless.
\end{enumerate} 
As Zemlyakov shows in his delightful article \cite{Z}, several important questions for such systems can be answered using the graphs of the molecule positions as functions of time. Following this, the two types of collision we consider are shown in Figure \ref{fig:Types}.
\begin{figure}
    \subfloat
{
\begin{tikzpicture}
[scale=.6]
\draw   [help lines, ->] (0,0)  -- (0,6);
\coordinate [label=right:${t}$]  (t) at (6,0);
\draw   [help lines, ->] (0,0)  -- (6,0);
\coordinate [label=above:${x}$]  (x) at (0,6);

\coordinate [label=left:${x_1(t_0)}$]  (x0) at (0,1);
\coordinate [label=left:${x_2(t_0)}$]  (x3) at (0,5);

\node at (0,1)  {.};
\node at (0,5)  {.};
\node at (0,0)  {.};

\draw  [lightgray,ultra thick] (0, 1)--(2,3);
\draw  [ultra thick, ->] (2, 3)--(4.3,5.3);
\draw  [ultra thick] (0, 5)--(2,3);
\draw  [lightgray,ultra thick, ->] (2, 3)--(4.3,.7);
\end{tikzpicture}
}
\qquad \    
\subfloat
{
\begin{tikzpicture}
[scale=.65]
\draw   [help lines, ->] (0,0)  -- (0,6);
\coordinate [label=right:${t}$]  (t) at (6,0);
\draw   [help lines, ->] (0,0)  -- (6,0);
\coordinate [label=above:${x}$]  (x) at (0,6);

\coordinate [label=left:${x_1(t_0)}$]  (x0) at (0,1);
\coordinate [label=left:${x_2(t_0)}$]  (x2) at (0,3);
\coordinate [label=left:${x_3(t_0)}$]  (x3) at (0,5);

\node at (0,1)  {.};
\node at (0,3)  {.};
\node at (0,5)  {.};
\node at (0,0)  {.};

\draw  [lightgray,ultra thick] (0, 1)--(2,3);
\draw  [ultra thick, ->] (2, 3)--(4.3,5.3);
\draw  [ultra thick,dashed, ->] (0, 3)--(4.3,3);
\draw  [ultra thick] (0, 5)--(2,3);
\draw  [lightgray,ultra thick, ->] (2, 3)--(4.3,.7);
\end{tikzpicture}
}
    \caption{The collisions of subsection \ref{sec: Zemlyakov}.}
    \label{fig:Types}
\end{figure}
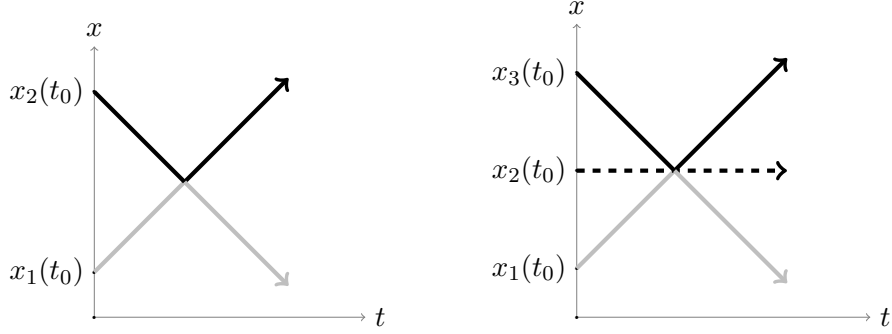

Consider a $1$-dimensional point system $\left(x_k^{(N)}(t), u_k^{(N)}(t)\right)$, $k=1,\ldots,N$ obeying elastic collision dynamics. Fix any $T\in (0,\infty)$. For all $t\in[0,T]$, assume that all collisions are binary or triple as above.
\begin{Prop} \label{Prop: Zemlyakov}
Let $S_t(x,v)=(x+vt,v)$. For all $t\in[0,T]$ the empirical measures
\begin{equation}
M_t(dx,dv)
=
\frac1N \sum_{k = 1}^N \delta_{\left(x_k^{(N)}(t), u_k^{(N)}(t)\right)}(dx,dv)
\end{equation}
satisfy
\begin{equation}
     M_t^{(N)}(dx,dv) = S_t M_0^{(N)}(dx,dv).
\end{equation}
\end{Prop}
\begin{proof} Merely notice that for each $t$
\begin{equation}
\begin{split}
     \frac1N \sum_{k = 1}^N \delta_{\left(x_k^{(N)}(t), u_k^{(N)}(t)\right)}
     =
     \frac1N \sum_{k = 1}^N \delta_{\left(x_k^{(N)}(0)+t u_k^{(N)}(0), u_k^{(N)}(0)\right)}
\end{split}
\end{equation}
since there is a bijection, if multiplicities are taken into account: 
\begin{equation}
\begin{split}
      \left\{\left(x_k^{(N)}(t), u_k^{(N)}(t)\right)\right\}\leftrightarrow
      \left\{\left(x_k^{(N)}(0)+t u_k^{(N)}(0), u_k^{(N)}(0)\right)\right\}.
\end{split}
\end{equation}
Indeed, the  exchange of velocities between the moving molecules of a collision establishes a bijection between the orbits before and after that collision. Iterating this finitely many times brings us back to the initial orbits given by $
\left(x_k^{(N)}(0)+t u_k^{(N)}(0), u_k^{(N)}(0)\right)$.
\end{proof}
The following Lemma will be used repeatedly.
\begin{Lem}\label{transport flow limit}
Suppose that
\begin{equation}
   M_t^{(N)}(d x,d v)=S_t M_0^{{N}}(d x,d v),\quad
   M_0^{(N)}(d x,d v)\Rightarrow M_0(d x,d v).
\end{equation}
Then
$M_t^{(N)}(dx,dv)\Rightarrow  S_t M_0(dx,d v)$.
\end{Lem}
\begin{proof}
Use the definitions of weak convergence and push forward under $S_t$.
\end{proof}
%
%

As it is standard that $M_t(d x,d v)=S_t M_0(d x,d v)$ solves weakly the free transport equation
\begin{equation}
\partial_t M_t  + v\partial_x M_t=0
\end{equation}
we shell refer to it as the {\em a free transport flow}.


\subsection{Euler system from free transport flow}\label{hydro equations}

We find here conditions that imply that averages with respect to free transport flow satisfy the compressible Euler system in dimension $1$.
The next two subsections provide examples satisfying such conditions.

\begin{Lem}\label {lem: macro eqn of g(v)}
Suppose that $M_t(dx,dv)=S_t M_0(dx,dv)$. Then for all $\phi(t,x)\in C_c^1\left( [0,T) \times \R \right)$ and $g(v)$ such that $v g(v) \in L^1 \left( M_0\right)$, we have
\begin{equation} \label {eq: general initial value eqn}
\begin{split}
    \int_0^T \int_{\R^{2}} 
    [
    \partial_t\phi(t, x) g(v)
    +&
    \partial_x \phi(t,x)\, vg(v)
     ]
     M_t(d x, dv )dt\\
    &\quad\quad\quad\quad\quad\quad+
    \int_{\R^{2}} \phi(0,x) g(v) M_0(dx,dv)
    =
    0.
\end{split}
\end{equation}
\end{Lem}
\begin{proof}
Straight forward calculation using the definition of the push forward under $S_t$ and the assumption that $\phi$ is compactly supported.
\end{proof}
Disintegrating $M_t(dx,dv)$ of Lemma \ref{lem: macro eqn of g(v)} as 
 \begin{equation}
 \displaystyle M_t(dx,dv) = \int M_{t,x}(dv) \mu_t(dx),
 \end{equation}
and for
  \begin{equation}
 \overline{g(v)}(t,x) =  \int g(v) M_{t,x}(dv),
 \end{equation}
 \eqref{eq: general initial value eqn} becomes 
\begin{equation} \label {eq: general initial value eqn in bar form}
\begin{split}
    \int_0^T \int_{\R}  
    &\left[
    \partial_t\phi(t,x) \, \overline{g(v)}(t,x)
    +
    \partial_x \phi(t,x) \, \overline{vg(v)}(t,x)
     \right]
     \mu_t(d x)dt\\
    &\quad\quad\quad\quad\quad\quad\quad\quad\quad\quad\quad\quad+
    \int_{\R } \phi(0,x)\, \overline{g(v)}(0,x) \mu_0 (dx)
    =
    0.
\end{split}
\end{equation}
To apply Lemma  \ref {lem: macro eqn of g(v)} for $g(v)=1,v$, and $\dfrac12 v^2$, assume $v^3 \in L^1 \left( M_0\right)$.
Noting that
\begin{equation}
\displaystyle u(t,x) =\overline v(t,x)= \int v M_{t,x}(dv),
\end{equation}
and using the notation
\begin{equation} \label{eq: xi^2 xi^3}
\overline {\xi^2}(t,x)
=
 \int_{\R} (v-u(t,x))^2 M_{t,x}(dv),
 \quad
 \overline {\xi^3}(t,x)
=
 \int_{\R} (v-u(t,x))^3 M_{t,x}(dv),
\end{equation}
it follows that 
\begin{equation} \label {eq: decomposition v v^2 v^3}
\begin{split}
    & \overline{v^2}(t,x) =  u^2(t,x) + \overline{\xi^2} (t,x),\\
    & \overline{v^3}(t,x) =  u^3(t,x) 
    + 3u(t,x) \overline{\xi^2} (t,x) + \overline{\xi^3}(t,x).
\end{split}
\end{equation}
Then  \eqref{eq: general initial value eqn in bar form} for $\displaystyle g(v) = 1, v,$  and $\dfrac 12 v^2$ gives
\begin{equation} \label{eq: general initial value eqn in bar form and decomposed}
\begin{split}
    &\int_0^T \int_{\R} 
    ( \partial_t\phi  +  \partial_x \phi u )
     \mu_t(d x)dt
     +
    \int_{\R} \phi(0,x) \mu_0(dx)
    =
    0,\\
    &\int_0^T \int_{\R} 
    \left( \partial_t\phi u  +  \partial_x\left( \phi u^2  + \overline{\xi^2} \right) \right)
     \mu_t(d x)dt
     +
    \int_{\R} \phi(0,x) u\, \mu_0(dx)
    =
    0,\\
    &\int_0^T \int_{\R} 
   \left\{
     \partial_t\phi \left(\dfrac 12 u^2 + \dfrac 12 \overline{\xi^2}\right)
       +
     \partial_x \phi \left[ \left( \dfrac 12 u^2  + \dfrac 32 \overline{\xi^2} \right) u +\dfrac {\overline{\xi^3}}{2} \right ]
     \right\}
     \mu_t(d x)dt\\
     &\quad\quad\quad\quad\quad\quad\quad\quad\quad\quad\quad\quad\quad\quad
     +
    \int_{\R} \phi(0,x) \left(\dfrac 12 u^2(0,x) + \dfrac 12 \overline{\xi^2}(0,x)\right)  \mu_0(dx)
    =
    0.\\
\end{split}
\end{equation}
Moreover, if $\mu_t(dx) = \rho(t,x)dx$, $\overline{\xi^3}(t,x)=0$ and for $e(t,x) = \dfrac{ \overline{\xi^2}(t,x)}{2}$, $p=2\rho e$,  \eqref{eq: general initial value eqn in bar form and decomposed} shows that $\rho, u, e$ solve weakly the Cauchy problem
 \begin{equation} \label{cEic with internal energy}
\left\{
\begin{array}{l}
        \partial_t\rho+\partial_x(\rho u)=0\\
        \partial_t(\rho u)+\partial_x\left(\rho        
        u^2\right)+\partial_x p=0\\
        
        \partial_t
        \left(\rho \dfrac{u^2}2 + 
        \rho e\right)
        +
        \partial_x 
        \left( \rho u
        \left(
        \dfrac{u^2}2 + 
        e 
        \right)+pu
        \right)
        =0,\\
        p = 2\rho e,\\
        
        \rho\vert_{t=0} = \rho(0,x), \quad
        u\vert _{t=0} = u(0,x), \quad 
        e\vert_{t=0} = e(0,x),
\end{array}
\right.
\end{equation}
the one dimensional Euler system, cf.\ \cite{CF}, p.\,7. In summary, we have shown:
\begin{Prop}\label{prop: euler from free flow}
For $M_t(dx,dv)=S_t M_0(dx,dv)$, suppose that $v^3 \in L^1 \left( M_0\right)$, $\mu_t(dx) = \rho(t,x)dx$, and $\overline{\xi^3}(t,x)=0$. Then $\rho(t,x), u(t,x), e(t,x)$ as defined above is a weak solution to the one dimensional Euler system \eqref{cEic with internal energy}.
\end{Prop}
{The definition of initial conditions for weak solutions here is compatible with the one in \cite{dP}, p.\,2 and \cite{VF}, \S VII.10.}
Two examples satisfying the conditions of this proposition now follow.

\subsection{Two-layer system}   \label{sec: isentropic molecular}

For $N$ fixed, consider $N=2n$ point molecules $x_1$, $x_2$,\ldots, $x_{N}$ on the real line, with
\begin{equation} \label{eq: ic for 2-layer}
\begin{split}
   x_k(0)=\dfrac{k}{N},
   \quad 
   u_k(0)=
   \left\{
   \begin{array}{cc}
   1& \text{for $k$ odd}\\
   -1& \text{for $k$ even}.
   \end{array}
   \right.
\end{split}
\end{equation}
Let the system evolve as in subsection \ref{sec: Zemlyakov}. After the first $n$ simultaneous collisions take place the molecules with labels $1$ and ${N}$  move with velocities $1$ and ${-1}$, respectively, without ever interacting with any other molecule again.
The remaining molecules now form a replica of the initial system, reduced by two molecules.
 
As in \cite{Z}, the graphs of the positions as functions of time show the evolution of the system, Figure \ref{fig:2-layer}. 
\begin{figure}
\begin{tikzpicture}[scale=.3]
\draw   [help lines, ->] (0,0)  -- (0,20);
\coordinate [label=right:${t}$]  (t) at (20,0);
\draw   [help lines, ->] (0,0)  -- (20,0);
\coordinate [label=above:${x}$]  (x) at (0,20);
\coordinate [label=left:${x_1(0)}$]  (x0) at (0,0);
\coordinate [label=left:${x_{N}(0)}$]  (xN) at (0,18);
\node at (0,0)  {.};
\node at (0,2)  {.};
\node at (0,4)  {.};
\node at (0,6)  {.};
\node at (0,8)  {.};
\node at (0,10)  {.};
\node at (0,12)  {.};
\node at (0,14)  {.};
\node at (0,16)  {.};
\node at (0,18)  {.};

\draw  [dotted,ultra thick] (0, 0)--(1,1);
\draw  [ultra thick] (1, 1)--(3,3);
\draw  [dotted,ultra thick] (3, 3)--(5,5);
\draw  [ultra thick] (5, 5)--(7,7);
\draw  [dotted,ultra thick] (7, 7)--(9,9);
\draw  [ultra thick, ->] (9, 9)--(18,18);

\draw  [dotted,ultra thick] (0, 4)--(1,5);
\draw  [ultra thick] (1, 5)--(3,7);
\draw  [dotted,ultra thick] (3, 7)--(5,9);
\draw  [ultra thick] (5, 9)--(7,11);
\draw  [dotted,ultra thick, ->] (7, 11)--(14,18);

\draw  [dotted,ultra thick] (0, 8)--(1,9);
\draw  [ultra thick] (1,9)--(3,11);
\draw  [dotted,ultra thick] (3,11)--(5,13);
\draw  [ultra thick, ->] (5, 13)--(10,18);

\draw  [dotted,ultra thick] (0, 12)--(1,13);
\draw  [ultra thick] (1, 13)--(3,15);
\draw  [dotted,ultra thick,->] (3, 15)--(6,18);

\draw  [dotted,ultra thick] (0, 16)--(1,17);
\draw  [ultra thick, ->] (1, 17)--(2,18);

\draw  [ultra thick] (0, 2)--(1,1);
\draw  [dotted,ultra thick, ->] (1, 1)--(2,0);

\draw  [ultra thick] (0, 6)--(1,5);
\draw  [dotted,ultra thick] (1, 5)--(3,3);
\draw  [ultra thick, ->] (3, 3)--(6,0);
 
\draw  [ultra thick] (0, 10)--(1,9);
\draw  [dotted,ultra thick] (1,9)--(3,7);
\draw  [ultra thick] (3, 7)--(5,5);
\draw  [dotted,ultra thick, ->] (5, 5)--(10,0);
 
\draw  [ultra thick] (0, 14)--(1,13);
\draw  [dotted,ultra thick] (1, 13)--(3,11);
\draw  [ultra thick] (3, 11)--(5,9);
\draw  [dotted,ultra thick] (5, 9)--(7,7);
\draw  [ultra thick, ->] (7, 7)--(14,0);

\draw  [ultra thick] (0, 18)--(1,17);
\draw  [dotted,ultra thick] (1, 17)--(3,15);
\draw  [ultra thick] (3, 15)--(5,13);
\draw  [dotted,ultra thick] (5, 13)--(7,11);
\draw  [ultra thick] (7, 11)--(9,9);
\draw  [dotted,ultra thick, ->] (9, 9)--(18,0);
 
\end{tikzpicture}
\caption{Microscopic evolution of subsection \ref{sec: isentropic molecular}.} \label{fig:2-layer}
\end{figure}
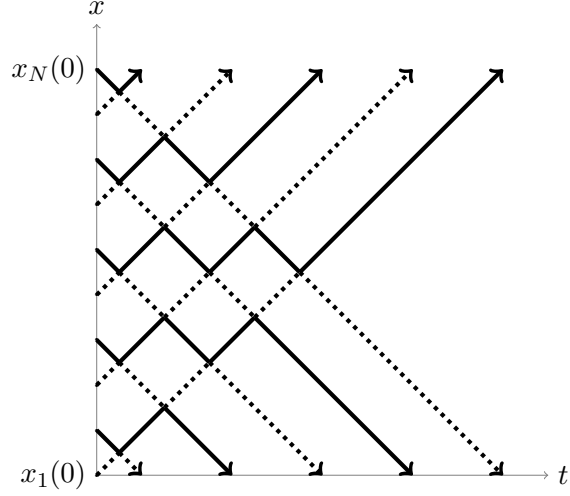 
For 
\begin{equation}
       M_t^{(N)}
        =
        \dfrac 1{N}
        \sum_{k=1}^{N} \delta_{\left(x_k(t),u_k(t)\right)},
\end{equation}
according to Proposition \ref{Prop: Zemlyakov},
\begin{equation}  
\begin{split}
       M_t^{(N)} = S_t M_0^{(N)}. 
\end{split}
\end{equation}
On the other hand, it is easy to check that as $N \to \infty$, 
\begin{equation}
\begin{split}
     M_0^{(N)}(dx,dv) 
     \Rightarrow 
     M_0(dx,dv)
        &=
        \chi_{[0,1]}(x)dx
        \otimes
        \left(
        \dfrac 12 \delta_{-1}(dv)+\dfrac 12 \delta_{1}(dv)
        \right), \end{split}
\end{equation}
therefore, by Lemma \ref{transport flow limit},
\begin{equation}
\begin{split}
     M_t^{(N)} 
     \Rightarrow 
     M_t=S_t M_0
     , 
     \quad N \to \infty. 
\end{split}
\end{equation}
It is straightforward to calculate that
\begin{equation}
\begin{split}
    M_t(dx,dv) 
    = 
    \frac12 \chi_{[t, t+1]}(x) dx\otimes \delta_1(dv)
    +
    \frac12 \chi_{[-t, -t+1]}(x) dx\otimes \delta_{-1}(dv).
\end{split}
\end{equation}
$M_t$ describes two layers, each of total mass $1/2$, initially overlapping on the interval $[0,1]$, moving with velocities $\pm1$ for $t\geq 0$, see Figure \ref{Fig: two vertical leayers}.
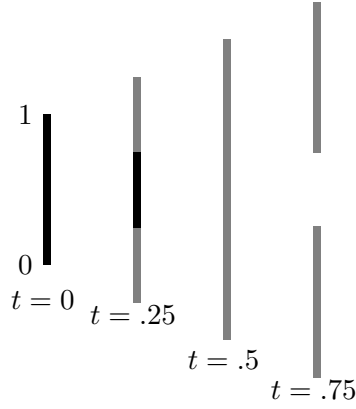
\begin{figure}
%
%
%
%
%
%
%
%
\begin{tikzpicture}
[scale=.2]


\coordinate [label=left:${0}$]  (x0) at (0,0);
\coordinate [label=left:${1}$]  (x3) at (0,10);

\coordinate [label=below:${t=0}$]  (t0) at (0,-1);
\coordinate [label=below:${t=.25}$]  (t1) at (6,-2);
\coordinate [label=below:${t=.5}$]  (t2) at (12,-5);
\coordinate [label=below:${t=.75}$]  (t3) at (18,-7);

\fill (0,0)--(0,10)--(.5,10)--(.5,0);

\fill [gray](6,-2.5)--(6,2.5)--(6.5,2.5)--(6.5,-2.5);
\fill [gray](6,7.5)--(6,12.5)--(6.5,12.5)--(6.5,7.5);
\fill [](6,2.5)--(6,7.5)--(6.5,7.5)--(6.5,2.5);

\fill [gray]   (12,-5)--(12,15)--(12.5,15)--(12.5,-5);

\fill [gray] (18,-7.5)--(18,2.55)--(18.5,2.55)--(18.5,-7.5);
\fill [gray] (18,7.5)--(18,17.5)--(18.5,17.5)--(18.5,7.5);

\end{tikzpicture}
    \caption{Macroscopic evolution of subsection \ref{sec: isentropic molecular}.}
    \label{Fig: two vertical leayers}
\end{figure}
The macroscopic density, velocity and energy density given by $M_t$ are
\begin{equation}
\left\{
 \begin{array}{l}
         \rho(t,x)
         =\dfrac12 \chi_{[t,1+t]}(x) + \dfrac12 \chi_{[-t,1-t]}(x)
         \\
        u(t,x)
        = \chi_{[t,1+t]}(x)  - \chi_{[-t,1-t]}(x),
        \\
        e(t,x)
        = \dfrac12 \chi_{[-t,1-t]}(x) \cdot \chi_{[t,1+t]}(x).
  \end{array}
  \right.
\end{equation}
Notice that $\displaystyle \int_{\R^2} |v|^3 M_0(dx,dv)<\infty$ and 
\begin{equation}
\overline{\xi^3}
=
\int_{\R} (v-u(t,x))^3 M_{t,x}(dv)
=
0.
\end{equation}
Therefore, by Proposition \ref{prop: euler from free flow}, $(\rho,u,e)$ is a solution to the Euler system
 \begin{equation} \label{system: euler from 2 layers}
\left\{
\begin{array}{l}
        \partial_t\rho+\partial_x(\rho u)=0\\
        \partial_t(\rho u)+\partial_x\left(\rho        
        u^2\right)+\partial_x p=0\\
        \partial_t
        \left(\rho \dfrac{u^2}2 + 
        \rho e\right)
        +
        \partial_x 
        \left( \rho u
        \left(
        \dfrac{u^2}2 + 
        e 
        \right)+pu
        \right)
        =0,\\
        p = 2\rho e,\\      
        \rho\vert_{t=0} = \chi_{[0,1]}(x), \quad
        u\vert _{t=0} = 0, \quad 
        e\vert_{t=0} = \dfrac12 \chi_{[0,1]}(x) .
\end{array}
\right.
\end{equation}

\subsection{Three-layer system} \label{sec: 3-layers}
Consider now for each $N=3n$ a second system, consisting of $N$ molecules $x_1, x_2,\ldots, x_{N}$ on the real line with 
\begin{equation} \label{eq: 3-layer initial}
\begin{split}
   &x_k(0)=\dfrac{k}{N},
   \quad k = 1, \ldots, N, \\
   &u_k(0)=
   \left\{
   \begin{array}{lll}
   \sqrt{6}/2   & \text{for $k=3m-2$}&\\
   0     & \text{for $k = 3m-1$}&\\
   -\sqrt{6}/2  & \text{for $k = 3m$},&\quad m = 1, \ldots, n,
   \end{array}
   \right.
   \end{split}
\end{equation}
also evolving under elastic collisions as in section \ref{sec: Zemlyakov}. 

The evolution of the system initialized by \eqref{eq: 3-layer initial} is shown in Figure \ref{fig:3-layer}. 
\begin{figure}
\begin{tikzpicture}[scale=.45]
\draw   [help lines, ->] (0,0)  -- (0,15);
\coordinate [label=right:${t}$]  (t) at (15,0);
\draw   [help lines, ->] (0,0)  -- (15,0);
\coordinate [label=above:${x}$]  (x) at (0,15);

\draw  [ultra thick] (0, 1)--(1,2);
\draw  [ultra thick, lightgray] (0, 2)--(1,2);
\draw  [ultra thick, dotted] (0, 3)--(1,2);

\draw  [ultra thick] (0, 4)--(1,5);
\draw  [ultra thick, lightgray] (0, 5)--(1,5);
\draw  [ultra thick, dotted] (0, 6)--(1,5);

\draw  [ultra thick] (0, 7)--(1,8);
\draw  [ultra thick, lightgray] (0, 8)--(1,8);
\draw  [ultra thick, dotted] (0, 9)--(1,8);

\draw  [ultra thick] (0, 10)--(1,11);
\draw  [ultra thick, lightgray] (0, 11)--(1,11);
\draw  [ultra thick, dotted] (0, 12)--(1,11);


\draw  [ultra thick] (1, 2)--(2,1);
\draw  [ultra thick, lightgray] (1, 2)--(2,2);
\draw  [ultra thick, dotted] (1, 2)--(2,3);

\draw  [ultra thick] (1,5)--(2,4);
\draw  [ultra thick, lightgray] (1, 5)--(2,5);
\draw  [ultra thick, dotted] (1, 5)--(2,6);

\draw  [ultra thick] (1,8)--(2,7);
\draw  [ultra thick, lightgray] (1, 8)--(2,8);
\draw  [ultra thick, dotted] (1,8)--(2,9);

\draw  [ultra thick] (1, 11)--(2,10);
\draw  [ultra thick, lightgray] (1, 11)--(2,11);
\draw  [ultra thick, dotted] (1, 11)--(2,12);


\draw  [ultra thick, ->] (2,1)--(3, 0);
\draw  [ultra thick, lightgray] (2,2)--(2.5,2);
\draw  [ultra thick, dotted] (2,3)--(2.5,3.5);

\draw  [ultra thick] (2,4)--(2.5,3.5);
\draw  [ultra thick, lightgray] (2,5)--(2.5, 5);
\draw  [ultra thick, dotted] (2,6)--(2.5,6.5);

\draw  [ultra thick] (2,7)--(2.5,6.5);
\draw  [ultra thick, lightgray] (2,8)--(2.5,8);
\draw  [ultra thick, dotted] (2,9)--(2.5,9.5);

\draw  [ultra thick] (2,10)--(2.5,9.5);
\draw  [ultra thick, lightgray] (2,11)--(2.5,11);
\draw  [ultra thick, dotted, ->] (2,12)--(3, 13);


\draw  [ultra thick, lightgray] (2,2)--(4,2);
\draw  [ultra thick, dotted] (2.5,3.5)--(4,2);

\draw  [ultra thick] (2.5,3.5)--(4,5);
\draw  [ultra thick, lightgray] (2,5)--(4, 5);
\draw  [ultra thick, dotted] (2.5,6.5)--(4, 5);

\draw  [ultra thick] (2.5,6.5)--(4,8);
\draw  [ultra thick, lightgray] (2.5,8)--(4,8);
\draw  [ultra thick, dotted] (2.5,9.5)--(4,8);

\draw  [ultra thick] (2.5,9.5)--(4,11);
\draw  [ultra thick, lightgray] (2.5,11)--(4,11);


\draw  [ultra thick, lightgray, ->] (4,2)--(6,0);
\draw  [ultra thick, dotted] (4,2)--(5.5,2);

\draw  [ultra thick] (4,5)--(5.5,3.5);
\draw  [ultra thick, lightgray] (4, 5)--(5.5, 5);
\draw  [ultra thick, dotted] (4, 5)--(5.5,6.5);

\draw  [ultra thick] (4,8)--(5.5, 6.5);
\draw  [ultra thick, lightgray] (4,8)--(5.5,8);
\draw  [ultra thick, dotted] (4,8)--(5.5,9.5);

\draw  [ultra thick] (4,11)--(5.5,11);
\draw  [ultra thick, lightgray, ->] (4,11)--(6,13);


\draw  [ultra thick, dotted] (5.5,2)--(7, 2);

\draw  [ultra thick] (5.5,3.5)--(7,2);
\draw  [ultra thick, lightgray] (5.5, 5)--(7,5);
\draw  [ultra thick, dotted] (5.5,6.5)--(7,5);

\draw  [ultra thick] (5.5, 6.5)--(7,8);
\draw  [ultra thick, lightgray] (5.5,8)--(7,8);
\draw  [ultra thick, dotted] (5.5,9.5)--(7,11);

\draw  [ultra thick] (5.5,11)--(7,11);


\draw  [ultra thick, dotted, ->] (7, 2)--(9,0);

\draw  [ultra thick] (7,2)--(10,2);

\draw  [ultra thick, lightgray] (7,5)--(10,2);
\draw  [ultra thick, dotted, ->] (7,5)--(12,5);

\draw  [ultra thick, ->] (7,8)--(12,8);
\draw  [ultra thick, lightgray] (7,8)--(10, 11);
\draw  [ultra thick, dotted] (7,11)--(10,11);

\draw  [ultra thick, ->] (7,11)--(9,13);


\draw  [ultra thick, ->] (10,2)--(12,0);

\draw  [ultra thick, lightgray, ->] (10,2)--(12,2);

\draw  [ultra thick, lightgray, ->] (10, 11)--(12,11);
\draw  [ultra thick, dotted, ->] (10,11)--(12, 13);

%
%


\end{tikzpicture}
\caption{ Microscopic evolution of subsection \ref{sec: 3-layers}.} \label{fig:3-layer}
\end{figure}
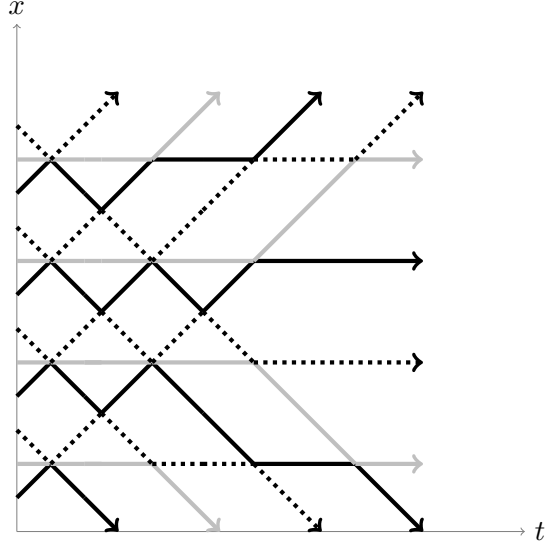
Again, if for the current system 
\begin{equation}
        \widetilde{M}_t^{(N)}
        =
        \dfrac 1{N}
        \sum_{k=1}^{N} \delta_{\left(x_k(t),u_k(t)\right)},
\end{equation}
by Proposition \ref{Prop: Zemlyakov},
\begin{equation}
\begin{split}
       \widetilde{M}_t^{(N)} = S_t \widetilde{M}_0^{(N)}. 
\end{split}
\end{equation}
On the other hand, as $N \to \infty$,
\begin{equation}
\begin{split}
     \widetilde{M}_0^{(N)}(dx,dv) 
     &\Rightarrow
     \widetilde{M}_0(dx,dv)\\
        &=
        \chi_{[0,1]}(x)dx
        \otimes
        \left(
        \dfrac 13 \delta_{- \sqrt 6/2}(dv)+\dfrac13\delta_0(dv)+\dfrac 13 \delta_{\sqrt6 /2}(dv)
        \right),
\end{split}
\end{equation}
therefore
\begin{equation}
     \widetilde{M}_t^{(N)} 
     \Rightarrow 
     \widetilde{M}_t = S_t \widetilde{M}_0, \quad N \to \infty.
\end{equation}
It is again a straightforward calculation that
\begin{equation}
\begin{split}
     \widetilde M_t(dx,dv)
     =
     &\frac13 \chi_{\left[-\frac{\sqrt{6}}{2}t,\ 1-\frac{\sqrt{6}}{2}t\right]}(x)
     dx
     \otimes
     \delta_{-\frac{\sqrt{6}}{2}}(dv) \\ 
     & +
     \frac13 \chi_{[0,1]}(x) dx
     \otimes
     \delta_{0}(dv)
     +
     \frac13 \chi_{\left[\frac{\sqrt{6}}{2}t,\ 1+\frac{\sqrt{6}}{2}t\right]}(x)
     dx
     \otimes
     \delta_{\frac{\sqrt{6}}{2}}(dv).
\end{split}
\end{equation} 
$\widetilde M_t$ describes three layers, each of total mass $1/3$, initially overlapping on the interval $[0,1]$. Two of them move with velocities $\pm\sqrt{6}/2$ for $t>0$, while the third stays at rest, see Figure \ref{Fig: three vertical leayers}.
\begin{figure}
\begin{tikzpicture}
[scale=.2]

\coordinate [label=left:${0}$]  (x0) at (0,0);
\coordinate [label=left:${1}$]  (x3) at (0,10);

\fill (0,0)--(0,10)--(.5,10)--(.5,0);

\fill [lightgray](6,-2.5)--(6,0)--(6.5,0)--(6.5,-2.5);
\fill [gray](6,0)--(6,2.5)--(6.5,2.5)--(6.5,0);
\fill (6,2.5)--(6,7.5)--(6.5, 7.5)--(6.5,2.5);
\fill [gray](6,7.5)--(6,10)--(6.5,10)--(6.5,7.5);
\fill [lightgray](6,10)--(6,12)--(6.5,12)--(6.5,10);

\fill[lightgray](12.5,-5)--(12.5,0)--(13,0)--(13,-5);
\fill [gray] (12.5,0)--(12.5,5)--(13, 5)--(13,0);
\fill [gray] (12.5,5)--(12.5,10)--(13, 10)--(13,5);
\fill [lightgray] (12.5,10)--(12.5,15)--(13, 15)--(13,10);

\fill [lightgray] (19, -7.5)--(19,0)--(19.5,0)--(19.5,-7.5);
\fill [gray] (19, 0)--(19,2.55)--(19.5,2.55)--(19.5,0);
\fill[lightgray]  (19, 2.55)--(19,7.5)--(19.5,7.5)--(19.5,2.55);
\fill [gray] (19, 7.5)--(19,10)--(19.5,10)--(19.5,7.5);
\fill [lightgray] (19, 10)--(19,17.5)--(19.5,17.5)--(19.5,10);

\coordinate [label=below:${t=0}$]  (t0) at (0,-1);
\coordinate [label=below:${t=.2}$]  (t1) at (6.5,-3);
\coordinate [label=below:${t=.4}$]  (t2) at (13,-5);
\coordinate [label=below:${t=.6}$]  (t3) at (19.5,-7);

\end{tikzpicture}
\caption
		{Macroscopic evolution of subsection \ref{sec: 3-layers}.}
\label
		{Fig: three vertical leayers}
\end{figure}
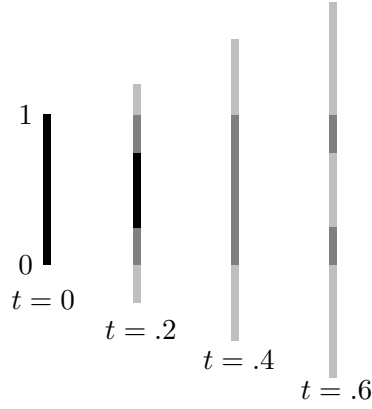
The macroscopic density, velocity and energy density given by $\widetilde M_t$ are
\begin{equation}
\left\{
 \begin{array}{l}
         \widetilde{\rho}(t,x)
         =\frac13 \chi_{\left[-\frac{\sqrt{6}}{2}t,\ 1-\frac{\sqrt{6}}{2}t\right]}(x) 
          +
          \frac13 \chi_{[0,1]}(x)
          +
         \frac13 \chi_{\left[\frac{\sqrt{6}}{2}t,\ 1+\frac{\sqrt{6}}{2}t\right]}(x)
         \\
        \widetilde{u}(t,x)
        = \dfrac {
        -\frac{\sqrt{6}}{6} \chi_{\left[-\frac{\sqrt{6}}{2}t,\ 1-\frac{\sqrt{6}}{2}t\right]}(x) 
          +
         \frac{\sqrt{6}}{6} \chi_{\left[\frac{\sqrt{6}}{2}t,\ 1+\frac{\sqrt{6}}{2}t\right]}(x)
       }  
        {\widetilde \rho (t,x) } 
        \\
       \widetilde e(t,x)
        =
        \dfrac{
         \dfrac14 \chi_{\left[-\frac{\sqrt{6}}{2}t,\ 1-\frac{\sqrt{6}}{2}t\right]}(x) 
         +
         \dfrac14 \chi_{\left[\frac{\sqrt{6}}{2}t,\ 1+\frac{\sqrt{6}}{2}t\right]}(x)
         -
         \dfrac 12 \widetilde\rho(t,x)\widetilde u^2(t,x)
         }
         {\widetilde{\rho}(t,x)}.
  \end{array}
  \right.
\end{equation}
When $\widetilde\rho(t,x)=0$, take $\widetilde u(t,x),\widetilde e(t,x)=0$.
Notice that 
\begin{equation}
\overline{\xi^3} \tilde{\, \phantom l} (0,x) 
=
\int_{\R} (v-\widetilde u(t,x))^3 \widetilde M_{t,x}(dv)
=
0.
\end{equation}
By Proposition \ref{prop: euler from free flow},  $\left(\widetilde \rho, \widetilde u, \widetilde e \right)$ is also a solution to the Cauchy problem \eqref{system: euler from 2 layers}, clearly distinct from the solution $\left(\rho, u, e \right)$. 

\begin{Rmk}
It is well known that weak solutions to systems like \eqref{system: euler from 2 layers} are not unique, see \cite{D}. This section provides a microscopic interpretation of such macroscopic non-uniqueness, showing that such phenomena are quite natural from a Hamiltonian point of view.


\end{Rmk}

\section*{Appendix : Motion in a Central Field}
\label{central field}
\setcounter{equation}{0}
\renewcommand\theequation{A.\arabic{equation}}
\setcounter{Thm}{0}
\renewcommand\theThm{A.\arabic{Thm}}
We establish some facts for the motion in dimension $2$ of a single particle in an external field of potential energy $\Phi$ of finite range $\sigma$:
\begin{equation} \label{eq: central}
\begin{split}
   \bold x''(t) = -\Phi'(|\bold x|)\frac{\bold x}{|\bold x|}.
\end{split}
\end{equation}
To accommodate  \eqref{Uinteraction}, assume that $\Phi:(0, \infty) \to [0,\infty)$ satisfies
\begin{equation}
   \lim_{r \to 0} \Phi (r) = + \infty, \quad
   \Phi' \leq 0, \quad \Phi'' \geq 0, \quad
   \Phi(r)  \neq 0 \Leftrightarrow  0< r<\sigma.
\end{equation}

Consulting Figure \ref{fig: complete collision}, let $O$ be the center of the potential $\Phi$.  A molecule $m$ enters the range of   $\Phi$ at $A$ with velocity $\bold v$ and leaves at $B$. For $D$ the middle of $AB$, the path of $m$ in the range of $\Phi$ is symmetric about $OD$, by the reversibility of the equations of motion. Decompose $\bold v(t)$ into $\bold v_1(t)$ and $\bold v_2(t)$ along $AB$ and $OD$, respectively, and let $E$ be the intersection of $OD$ and the trajectory of $m$. When $m$ crosses $OD$ it has moved $d$ on the direction of $OD$. If $\theta$ is the angle between $\bold v$ and $AB$ and $C$ is the point on $OD$ with $AC$ of direction $\bold v$, then
\begin{equation} \label{d<sigmasintheta}
   \begin{split}
       d=DE<CD=AC\cdot \sin\theta<AO\cdot \sin\theta=\sigma \sin \theta.
   \end{split}
\end{equation}
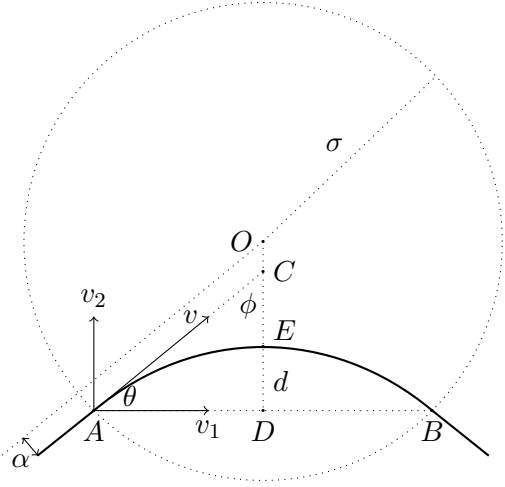
\begin{figure}
\begin{tikzpicture}
\coordinate [label=left:$O$]  (0) at (0,0);
  \coordinate [label=below:$A$] (A) at (-2.25,-2.25);
  \coordinate [label=right:$\theta$] (theta) at (-2.,-2.05);
  \draw[dotted] (0) let
              \p1 = ($ (A) - (0) $)
            in
              circle ({veclen(\x1,\y1)});
   \coordinate  [label=below:$B$] (B) at (2.25,-2.25);
   \draw [dotted](A) -- (B);
   \coordinate (V) at (2.3,2.2);
   \draw [dotted] (0) -- (V);
   \coordinate (W) at (-3.5,-2.87);
   \draw [dotted] (0) -- (W);
   \coordinate [label=below:$\sigma$] (sigma) at (.95,1.5);
   \coordinate [label=below:$\phi$] (phi) at (-.2,-.55);
   \node at (0,0) {.};
   \node at (2.25,-2.25) {.};
   \node at (-2.25,-2.25) {.};
   \coordinate [label=below:$D$] (D) at (0,-2.25);
   \node at (0,-2.25) {.};
   \draw [dotted](0) -- (D);
   \coordinate [label=above right:$E$] (E) at (0,-1.45);
   \node at (0,-1.40) {.}; 
   \coordinate [label=left:$v$] (v) at (-.725,-1.);
   \draw [->] (A) -- (v);
   \coordinate [label=below:$v_1$] (v_1) at (-.725,-2.25);
   \draw [->] (A) -- (v_1);
   \coordinate [label=above:$v_2$] (v_2) at (-2.25,-1);
   \draw [->] (A) -- (v_2);
   \coordinate [label=right:$d$]  (d) at (0,-1.85);
   \coordinate [label=right:$C$] (C) at (0,-.40);
   \node at (0,-.40) {.};
   \draw [dotted] (v) -- (C);
   \draw [thick](-2.25,-2.25) to [out=40,in=140] (2.25,-2.25);
   \draw [thick] (A) -- (-3,-2.85);
   \draw [thick] (B) -- (3,-2.85);
   \draw [<->] (-3.195,-2.61) -- (-2.99,-2.85);
   \coordinate [label=below:$\alpha$] (alpha) at (-3.22,-2.7);
\end{tikzpicture}
\caption{Motion in a central field of finite range $\sigma$.}
\label{fig: complete collision}
\end{figure}                                                                                                                                                                 
Let $T$ be the time it takes $m$ to travel from $A$ to $B$. 
\begin{Lem} 
\label{Lem: CollisionTimeEstimate}
For $\sigma$ the range of $\Phi$, $T$ and $v$ as above satisfy 
$T < \dfrac{4\sigma}{v}$.
\end{Lem}
 \begin{proof}      
From \eqref{eq: central}, 
\begin{equation}
\begin{split}
     v_2'' = 
     -\Phi''(|\bold x|)\frac{\bold x\cdot \bold x'}{|\bold x|^2}x_2
     -
     \Phi'(|\bold x|)\frac{ x'_2}{|\bold x|}
     +
     \Phi'(|\bold x|)\frac{\bold x\cdot \bold x'}{|\bold x|^3} x_2
     .
\end{split}
\end{equation}
For $x_2<0$ and as $\dfrac{d|\bold x|^2}{dt} <0$ for $t \in (0,T/2)$, and as $\Phi$ is convex, the first term of this is negative and, if $x_1$ is also negative,
the sum of the remaining two terms is also negative provided that
\begin{equation}
\begin{split}
    -
    { x'_2}{|\bold x|^2}
     +
    x_2 ({\bold x\cdot \bold x'}) 
     >0
     \Leftrightarrow
     -
    { x'_2}x_1 
     +
    x_2  x_1'  
     <0,
\end{split}
\end{equation}
since $\Phi$ is decreasing. Now note that $-{ x'_2}x_1 + x_2  x_1'$ stays constant in time and the inequality is satified at $t=0$.
Therefore $v_2$ is concave and by \eqref{d<sigmasintheta}
\begin{equation}
   \begin{split}
       \dfrac{v_2(0)}2 \cdot \dfrac T 2
        <d
       <\sigma \sin \theta, 
   \end{split}
\end{equation}
which, along with $v_2(0)=v\sin \theta$, concludes the proof.
\end{proof}
Still in Figure \ref{fig: complete collision}, let $\angle ACD=\phi$. 
Denoting the distance of $O$ from $AC$ (the impact parameter) by $\alpha$,
by \cite{LL}, p.\,49\footnote{Note here that \cite{LL}'s analysis of motion in a central field in their \S 14 is valid for any central field, including the ones with finite range.}


\begin{equation}  \label{phiangle}
    \begin{split}
         \phi(\alpha)
         =
      \int
         _{r_{min}}^{\infty}
         \dfrac{\alpha}{r^2\sqrt{1-\dfrac{\alpha^2}{r^2}-\dfrac{\Phi(r)}{E}}} dr,
    \end{split}
\end{equation}
where $E=\dfrac12 mv^2$ and $r_{min}$ is a zero of the radicand:
\begin{equation}\label{r_min}
    \begin{split}
    1-\dfrac{\alpha^2}{r_{min}^2}-\dfrac{\Phi(r_{min})}{E}=0.
    \end{split}
\end{equation}
\begin{Lem} 
For interaction potential as in \eqref{Uinteraction}, $r_{min} = r_{min}(\alpha)$ is increasing and $\phi(\alpha)$ is continuous on $[0,\infty)$.
\end{Lem}
\begin{proof}
For fixed $\alpha$ and $E$, the function
\begin{equation}
\begin{split}
     r \mapsto \dfrac{\alpha^2}{r^2} +\dfrac{\Phi(r)}{E}
\end{split}
\end{equation}
is strictly decreasing from $+\infty$ to $0$ for $r>0$ and the pre-image $r_{min}$ of $1$ satisfies \eqref{r_min}, or
\begin{equation}
\begin{split} \label{r_minalt}
    {\alpha}=\left( 1  -\dfrac{\Phi(r_{min})}{E}\right)^{1/2} {r_{min}}
\end{split}
\end{equation}
showing that $\alpha = \alpha(r_{min}
)$, and therefore $r_{min} = r_{min}(\alpha)$, is increasing.

To show that $\phi$ is continuous, change the variable in \eqref{phiangle} via $r=r_{min}y$:
\begin{equation}
  \begin{split}
         \phi(\alpha)
        =&      
         \int_1^{\infty}
         \dfrac 1
         {
         y^2\sqrt{
        \dfrac{ r_{min}^2}{\alpha^2}\left(1 - \dfrac {\Phi(r_{min}y)} {E}\right)-\dfrac{1}{y^2}
         }
         } 
         dy\\
         \substack{\text{by \eqref{r_minalt}}\\ \displaystyle =}
         &
        \int_1^{\infty}
         \dfrac 1
         {
         y^2\sqrt{
        \dfrac {E   -  \Phi(r_{min}y)} {E   -  \Phi(r_{min})}    -\dfrac{1}{y^2}
         }
         } 
         dy.
       \end{split}
\end{equation}
From \eqref{r_min} we have 
\begin{equation}
    \begin{split}
    E>\Phi(r_{min})
    \end{split}
\end{equation}
and since  $\Phi(r)$ is decreasing,
\begin{equation}
    \begin{split}
    \Phi(r_{min})\geq \Phi(r_{min}y), \quad y\geq 1,
    \end{split}
\end{equation}
therefore
\begin{equation}
   \begin{split}
         \dfrac 1
         {
         y^2\sqrt{
        \dfrac {E   -  \Phi(r_{min}y)} {E   -  \Phi(r_{min})}    -\dfrac{1}{y^2}
         }
         } 
         \leq
         \dfrac 1
         {
         y^2\sqrt{
        1    -\dfrac{1}{y^2}
         }
         }, 
       \end{split}
\end{equation}
{with} 
\begin{equation}
    \begin{split}
        \int_1^{\infty}
         \dfrac 1
         {
         y^2\sqrt{
        1    -\dfrac{1}{y^2}
         }
         } dy
         =\dfrac {\pi}2.
\end{split}
\end{equation}
In other words, the integrand of $\phi$ is dominated by an integrable function. This, and the continuity of $r_{min}$ in $\alpha$, show that $\phi$ is continuous in $\alpha$.
\end{proof}

\begin{Cor} \label{continuous dependence}
For any $0\leq\phi_0\leq  \pi/2$, there exists $0\leq \alpha_0 \leq \sigma$ such that $\phi(\alpha_0)=\phi_0$.
\end{Cor}
\begin{proof}
Just use continuity and that $\phi(0)=0$ (``head-on collision"), $\phi(\sigma)=\dfrac{\pi}{2}$ (no interaction).
\end{proof}

As is well known, motion in a central field also describes a system of two bodies interacting with each other via $\Phi$, a function of their distance, 
in a coordinate system with its origin at the center of mass of the system. The formulas for this transformation are in \cite{LL}, \S 13.

%




\begin{thebibliography}{dLS2}

\bibitem[AGS]{AGS}
Ambrosio L., N. Gigli, G. Savar\'e,
{\em Gradient Flows},
Birk\"auser, 
2005

\bibitem[B]{B}
Boltzmann, L.,
{\it Lectures on gas theory},  
University of California Press,
1964


\bibitem[BN]{BN}
Bressan, A., T. Nguyen,
{\em Non-existence and Non-uniqueness for Multidimensional Sticky Particle Systems},
Kinetic and Related Models,
{\bf 7} (2) (2014),
205--218.



\bibitem[CF]{CF} 
Courant, R., K.O. Friedrichs,
{\em Supersonic Flow and Shock Waves},
Interscience, 
1948

\bibitem[D]{D}
Dafermos, C.M.,
{\em Hyperbolic conservation laws in continuum physics}, 
vol. 325 of Grundlehren der Mathematischen Wissenschaften, Springer, 2000.

\bibitem[dLS]{dLS}
De\,Lellis, C., L. Sz\'ekelyhidi, 
{\em On admissibility criteria for weak solutions of the Euler equations},
Arch. Ration. Mech. Anal. {\bf 195} (1) (2010), 
225--260


\bibitem[dP]{dP}
DiPerna,~R.J., 
{\em Convergence of the viscosity method for isentropic gas dynamics},
Comm. Math. Phys. {\bf 91} (1) (1983),
1--30


\bibitem[DJX]{DJX}
Dostoglou, S., N.C. Jacob, Jianfei Xue,
{\em On hydrodynamic equations at the limit of 
infinitely many molecules},
J. Mathematical Sciences, New York {\bf 205} (2) (2015), 
91--104

\bibitem[ERS]{ERS}
E, Weinan, Yu.G. Rykov, Ya.G. Sinai,
{\em
Generalized variational principles, global weak solutions and behavior with random initial data for systems of conservation laws arising in adhesion particle dynamics
}
Comm. Math. Phys. Vol. {\bf 177}, (2) (1996), 349-380.



\bibitem[G]{G} Gal'perin,~G.A.,
{\em On Systems of Locally Interacting and Repelling Particles Moving in Space}, 
Trudy Moskov. Mat. Obshch. {\bf 43} (1981),
142--196




\bibitem[L]{L} 
Lanford, O.E., III, 
{\em Time evolution of large classical systems},  
Lecture Notes in Phys. {\bf 38}, 
Dynamical systems, theory and applications 
(Recontres, Battelle Res. Inst., Seattle, Wash., 1974),
Springer, Berlin, 1975,
1--111



\bibitem[LL]{LL} Landau, L.D., E.M. Lifshitz,
{\em Mechanics}, 
3rd Edition, Vol. 1 of Course of Theoretical Physics, 
Pergamon Press, 1976.


\bibitem[M]{M}
Maxwell, J.C., 
{\it On the Dynamical Theory of Gases}, 
Philosophical Transactions of the Royal Society of London {\bf 157} (1867),
49--88.

\bibitem[Mor]{Mor}
Morrey, C.B.,
{\em On the derivation of the Equations of Hydrodynamics from Statistical
Mechanics},  
Comm. Pure and Applied Math. {\bf VIII} (1955),
317--322.

\bibitem[N]{N}
Navier, C.L.,
{\em Sur Le Lois du Movement des Fluides},
M\'emoires, Academie des Sciences {\bf 6} (1827), 
389--440 (read 1822)



\bibitem[Sch]{Sch}
Scheffer, V., 
{\em An Inviscid Flow with Compact Support in Space-Time},
J. Geometric Analysis {\bf 3} (4) (1993),
343--401.

\bibitem[Sh]{Sh}
Shnirelman, A., 
{\em On the Nonuniqueness of Weak Solution of the Euelr Equation},
Comm. Pure Applied Math. {\bf L} (1997),
1261--1286 

%

\bibitem[V]{V} 
Vaserstein, L.N., 
{\em On Systems of Particles with Finite-Range and/or Repulsive Interactions}, 
Commun. Math. Phys. {\bf 69} (1979), 
31--56

\bibitem[VF]{VF}
Vishik, M.I., A.V. Fursikov,
{\em Mathematical Problems of Statistical Hydromechanics},
Kluwer 1988

\bibitem[X]{X} Xue, Jianfei,  
{University of Missouri Ph.D. Thesis}, 
in preparation. 

\bibitem[Z]{Z} 
Zemlyakov, A.N.,
{\em The arithmetic and geometry of collisions},
Kvant {\bf 4}, (1978), 14--21.

\end{thebibliography}
\end{document}